\documentclass[letter,runningheads]{llncs}

\usepackage{amsmath}
\usepackage{amssymb}
\usepackage{graphicx}
\usepackage{xspace}
\usepackage{enumerate}
\usepackage{comment}
\usepackage{cite}
\usepackage{algorithm}
\usepackage{times}
\usepackage[noend]{algpseudocode}
\usepackage[usenames]{xcolor}
\usepackage{stmaryrd} 
\usepackage{mathtools} 

\spnewtheorem{prop}{Property}{\bfseries}{\itshape} 

\newcommand{\remove}[1]{}

\renewenvironment{proof}
{{\em Proof.\ }}{\hspace*{\fill}$\Box$\par\vspace{2mm}}

\newcommand{\conf}{\otimes}

\begin{document}
\title{Advances on Testing C-Planarity of\\ Embedded Flat Clustered Graphs
\thanks{Research partially supported by the Australian Research Council (grant DE140100708).}
}
\author{Markus Chimani\inst{1} \and Giuseppe Di Battista\inst{2} \and Fabrizio Frati\inst{3} \and Karsten Klein\inst{3}}
\institute{Theoretical Computer Science, University Osnabr\"uck, Germany\\
\email{markus.chimani@uni-osnabrueck.de}
\and Dipartimento di Ingegneria, University Roma Tre, Italy\\
\email{gdb@dia.uniroma3.it}
\and School of Information Technologies, The University of Sydney, Australia\\
\email {\{fabrizio.frati,karsten.klein\}@sydney.edu.au}}
\maketitle

\begin{abstract}
We show a polynomial-time algorithm for testing $c$-planarity of embedded flat clustered graphs with at most two vertices per cluster on each face.
\end{abstract}

\section{Introduction}
A \emph{clustered graph} $C(G,T)$ consists of a graph $G(V,E)$, called {\em underlying graph}, and of a rooted tree $T$, called {\em inclusion tree}, representing a cluster hierarchy on $V$. The vertices in $V$ are the leaves of $T$, and the inner nodes of $T$, except for the root, are called \emph{clusters}. The vertices that are descendants of a cluster $\alpha$ in $T$ {\em belong to} $\alpha$ or {\em are in} $\alpha$. A \emph{$c$-planar drawing} of $C$ is a planar drawing of $G$ together with a representation of each cluster $\alpha$  as a simple connected region $R_{\alpha}$ enclosing all and only the vertices that are in $\alpha$; further, the boundaries of no two such regions $R_{\alpha}$ and $R_{\beta}$ intersect; finally, only the edges connecting vertices in $\alpha$ to vertices not in $\alpha$ cross the boundary of $R_{\alpha}$, and each does so only once. A clustered graph is \emph{$c$-planar} if it admits a $c$-planar drawing.

Clustered graphs find numerous applications in computer science~\cite{s-gc-07}, thus theoretical questions on clustered graphs have been deeply investigated. From the visualization perspective, the most intriguing question is to determine the complexity of testing $c$-planarity of clustered graphs. Unlike for other planarity variants~\cite{s-ttphtpv-13}, like {\em upward planarity}~\cite{GT01} and {\em partial embedding planarity}~\cite{adfjkpr-tppeg-10}, the complexity of testing $c$-planarity remains unknown since the problem was posed nearly two decades ago~\cite{fce-pcg-95}.


Polynomial-time algorithms to test the $c$-planarity of a clustered graph $C$ are known if $C$ belongs to special classes of clustered graphs~\cite{cw-cccg-06,cdpp-cccc-05,df-ectefcgsf-09,fce-pcg-95,d-ltarcgp-98,cdfpp-cccg-08,gls-cecg-05,gjlmpw-actcg-02,jkkpsv-scceg-09,jstv-cfoe-08}, including {\em $c$-connected clustered graphs}, that are clustered graphs $C(G,T)$ in which, for each cluster $\alpha$, the subgraph $G[\alpha]$ of $G$ induced by the vertices in $\alpha$ is connected~\cite{fce-pcg-95,d-ltarcgp-98,cdfpp-cccg-08}. Effective ILP formulations and FPT algorithms for testing $c$-planarity have been presented~\cite{cgjkm-ssscp-09,ck-ssscp-12}. Generalizations of the $c$-planarity testing problem have also been considered~\cite{addfpr-rccp-14,afp-scgc-09,dgl-ocp-08}.

An important variant of the $c$-planarity testing problem is the one in which the clustered graph $C(G,T)$ is {\em flat} and {\em embedded}. That is, every cluster is a child of the root of $T$ and a planar embedding for $G$ (an order of the edges incident to each vertex) is fixed in advance; then, the $c$-planarity testing problem asks whether a $c$-planar drawing exists in which $G$ has the prescribed planar embedding. This setting can be highly regarded for several reasons. First, several NP-hard graph drawing problems are polynomial-time solvable in the fixed embedding scenario, e.g., {\em upward planarity testing}~\cite{bdlm-udtd-94,GT01} and {\em bend minimization in orthogonal drawings}~\cite{t-eggmnb-87,GT01}. Second, testing $c$-planarity of embedded flat clustered graphs generalizes testing $c$-planarity of triconnected flat clustered graphs. Third, testing $c$-planarity of embedded flat clustered graphs is strongly related to a seemingly different problem, that we call {\em planar set of spanning trees in topological multigraphs} ({\sc pssttm}): Given a non-planar topological multigraph $A$ with $k$ connected components $A_1,\dots,A_k$, do spanning trees $S_1,\dots,S_k$ of $A_1,\dots,A_k$ exist such that no two edges in $\bigcup_i S_i$ cross? Starting from an embedded flat clustered graph $C(G,T)$, an instance $A$ of the {\sc pssttm} problem can be constructed that admits a solution if and only if $C(G,T)$ is $c$-planar: $A$ is composed of the edges that can be inserted inside the faces of $G$ between vertices of the same cluster, where each cluster defines a multigraph $A_i$. The {\sc pssttm} problem is NP-hard, even if $k=1$~\cite{kln-nstl-91}.


Testing $c$-planarity of an embedded flat clustered graph $C(G,T)$ is a polynomial-time solvable problem if $G$ has no face with more than five vertices and, more in general, if $C$ is a {\em single-conflict} clustered graph~\cite{df-ectefcgsf-09}, i.e., the instance $A$ of the {\sc pssttm} problem associated with $C$ is such that each edge has at most one crossing. A polynomial-time algorithm is also known for testing $c$-planarity of embedded flat clustered graphs such that the graph induced by each cluster has at most two connected components~\cite{jjkl-ecgtcc-09}. Finally, the $c$-planarity of clustered cycles with at most three clusters~\cite{cdpp-cccc-05} or with each cluster containing at most three vertices~\cite{jkkpsv-scceg-09} can be tested in polynomial time.

\subsubsection*{Contribution and outline.} In this paper we show how to test $c$-planarity in cubic time for embedded flat clustered graphs $C(G,T)$ such that at most two vertices of each cluster are incident to any face of $G$. While this setting might seem unnatural at a first glance, its study led to a deep (in our opinion) exploration of some combinatorial properties of highly non-planar topological graphs. Namely, every instance $A$ of the {\sc pssttm} problem arising from our setting is such that there exists no sequence $e_1,e_2,\dots,e_h$ of edges in $A$ with $e_1$ and $e_h$ in the same connected component of $A$ and with $e_i$ crossing $e_{i+1}$, for every $1\leq i\leq h-1$; these instances might contain a quadratic number of crossings, which is not the case for single-conflict clustered graphs~\cite{df-ectefcgsf-09}. Within our setting, performing all the ``trivial local'' tests and simplifications results in the rise of nice global structures, called {\em $\alpha$-donuts}, whose study was interesting to us.

The paper is organized as follows. In Section~\ref{se:preliminaries} we introduce some preliminaries; in Section~\ref{se:outline} we give an outline of our algorithm; in Section~\ref{se:algorithm} we describe our algorithm and prove its correctness; finally, in Section~\ref{se:conclusions} we conclude.

\section{Saturators, Con-Edges, and Spanning Trees} \label{se:preliminaries}

A natural approach to test $c$-planarity of a clustered graph $C(G(V,E),T)$ is to search for a {\em saturator} for $C$. A set $S\subseteq V\times V$ is a saturator for $C$ if $C'(G'(V,E\cup S),T)$ is a $c$-connected $c$-planar clustered graph. Determining the existence of a saturator for $C$ is equivalent to testing the $c$-planarity of $C$~\cite{fce-pcg-95}. Thus, the core of the problem consists of determining $S$ so that $G'[\alpha]$ is connected, for each $\alpha\in T$, and so that $G'$ is planar.

In the context of embedded flat clustered graphs (see Fig.~\ref{fig:con-edges}(a)), the problem of finding saturators becomes seemingly simpler. Since the embedding of $G$ is fixed, the edges in $S$ can only be embedded inside the faces of $G$, in order to guarantee the planarity of $G'$. This implies that, for any two edges $e_1$ and $e_2$ that can be inserted inside a face $f$ of $G$, it is known {\em a priori} whether $e_1$ and $e_2$ can be both in $S$, namely only if their end-vertices do not alternate along the boundary of $f$. Also, $S$ can be assumed to contain only edges connecting vertices that belong to the same cluster, as edges connecting vertices belonging to different clusters ``do not help'' to connect any cluster. For the same reason, $S$ can be assumed to contain only edges connecting vertices belonging to distinct connected components of $G[\alpha]$, for each cluster $\alpha$.

\begin{figure}[tb]
\begin{center}
\begin{tabular}{c c c c}
\mbox{\includegraphics[scale=0.315]{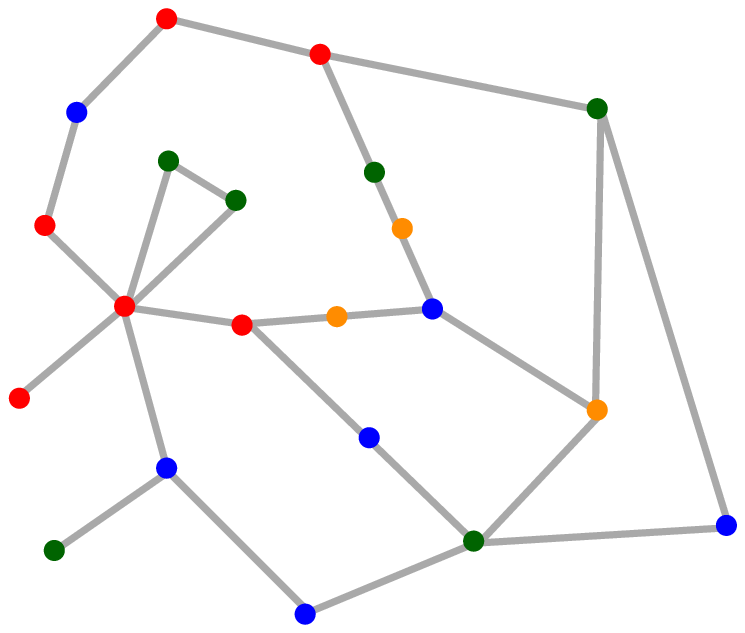}} \hspace{1mm} &
\mbox{\includegraphics[scale=0.315]{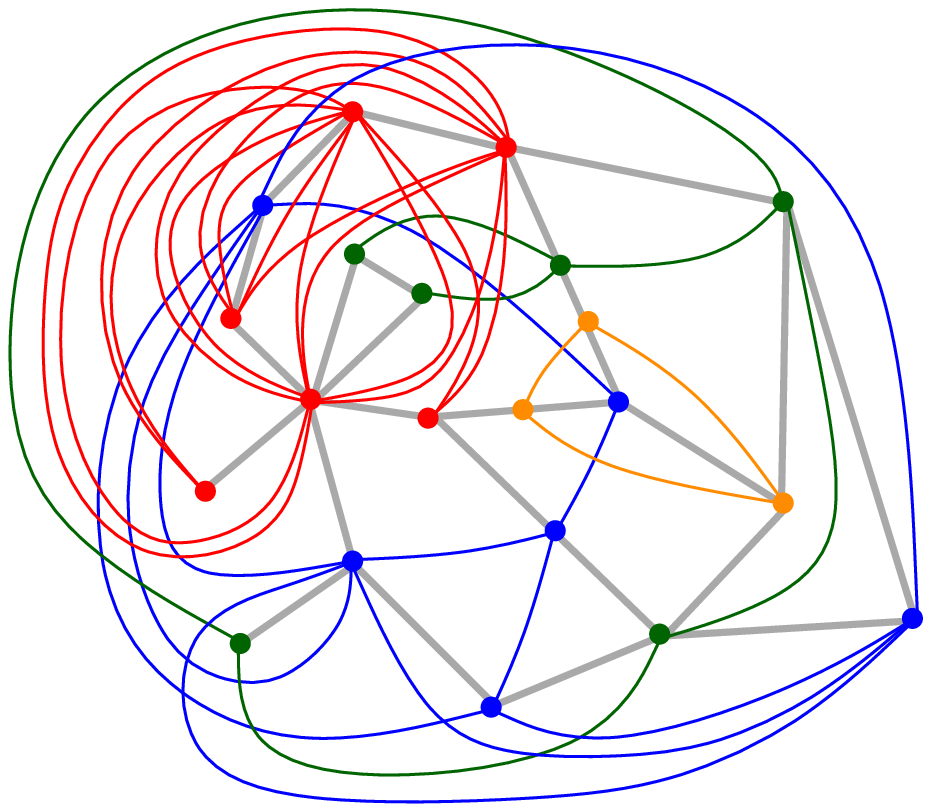}} \hspace{1mm} &
\mbox{\includegraphics[scale=0.315]{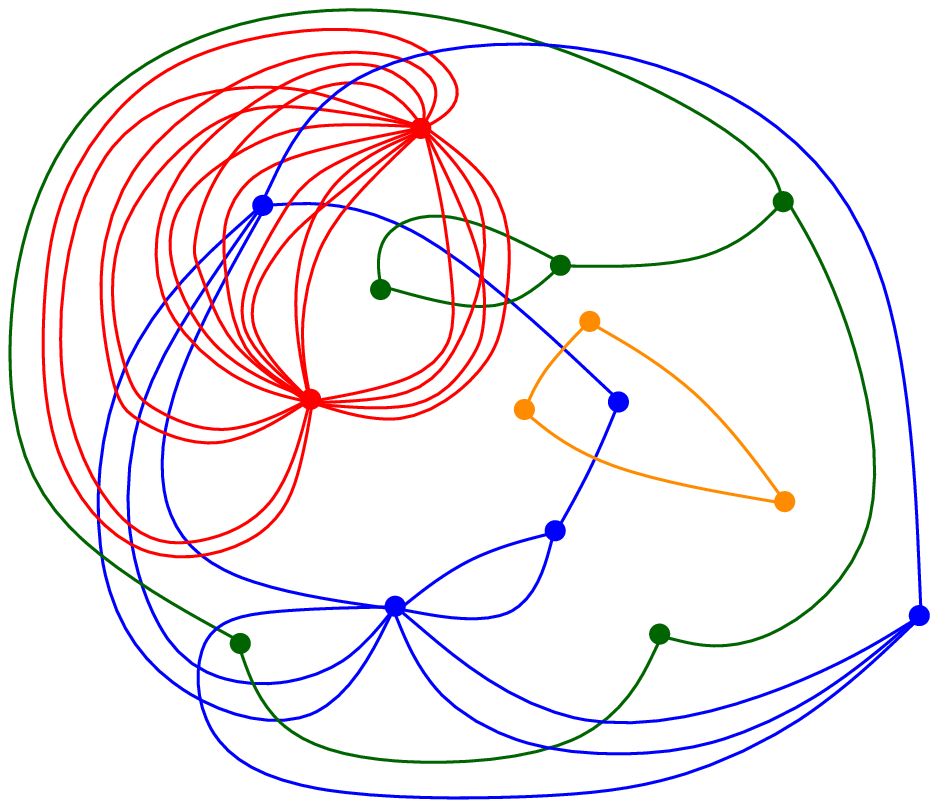}} \hspace{1mm} &
\mbox{\includegraphics[scale=0.315]{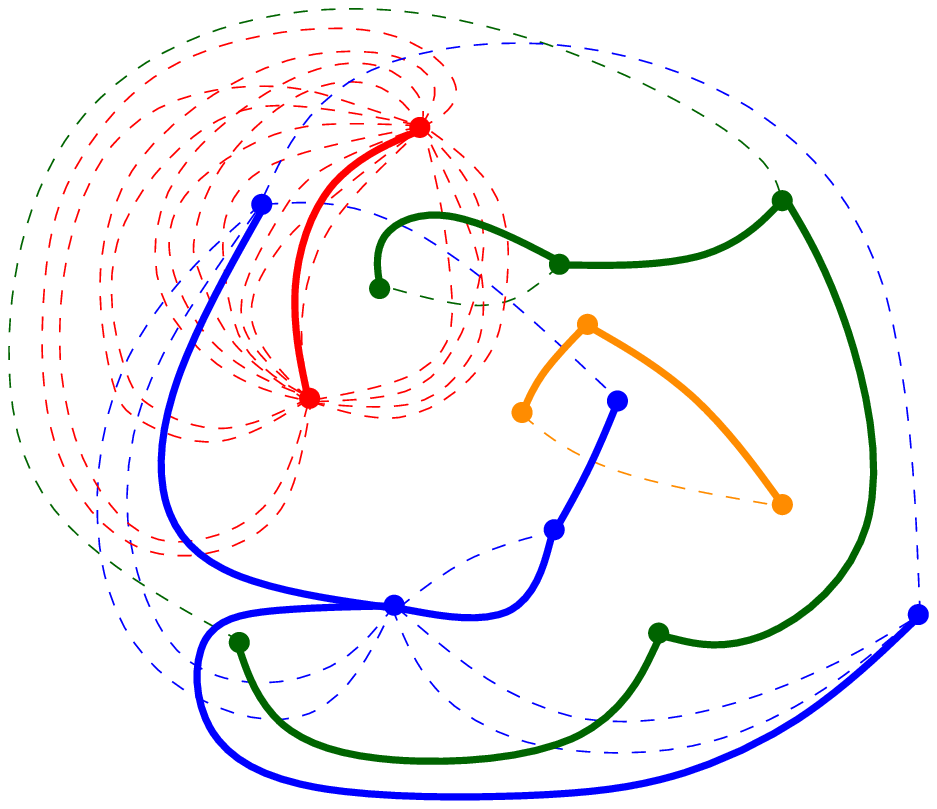}}\\
(a) \hspace{1mm} & (b) \hspace{1mm} & (c) \hspace{1mm} & (d)
\end{tabular}
\caption{(a) A clustered graph $C$. (b) Con-edges in $C$. (c) Multigraph $A$. (d) A planar set $S$ of spanning trees for $A$. Edges in $S$ are thick and solid, while edges in $A\setminus S$ are thin and dashed.}
\label{fig:con-edges}
\end{center}
\end{figure}

Consider a face $f$ of $G$ and let $B_f=(o_1,\dots,o_k)$ be the clockwise order of the occurrences of vertices along the boundary of $f$, where $o_i$ and $o_j$ might be occurrences of the same vertex $u$ (this might happen if $u$ is a cut-vertex of $G$). A {\em con-edge} (short for {\em connectivity-edge}) is a pair of occurrences $(o_i,o_j)$ of distinct vertices both belonging to a cluster $\alpha$, both incident to $f$, and belonging to different connected components of $G[\alpha]$ (see Fig.~\ref{fig:con-edges}(b)). If there are $\ell$ distinct pairs of occurrences of vertices $u$ and $v$ along a single face $f$, then there are $\ell$ con-edges connecting $u$ and $v$ in $f$, one for each pair of occurrences. A {\em con-edge for $\alpha$} is a con-edge connecting vertices in a cluster $\alpha$. Two con-edges $e$ and $e'$ in $f$ {\em have a conflict} or {\em cross} (we write $e \conf e'$) if the occurrences in $e$ alternate with the occurrences in $e'$ along the boundary of $f$.

The {\em multigraph $A$ of the con-edges} is an embedded multigraph that is defined as follows. Starting from $G$, insert all the con-edges inside the faces of $G$; then, for each cluster $\alpha$ and for each connected component $G_i[\alpha]$ of $G[\alpha]$, contract $G_i[\alpha]$ into a single vertex; finally, remove all the edges of $G$.  See Fig.~\ref{fig:con-edges}(c). With a slight abuse of notation, we denote by $A$ both the multigraph of the con-edges and the set of its edges. For each cluster $\alpha$, we denote by $A[\alpha]$ the subgraph of $A$ induced by the con-edges for $\alpha$. A {\em planar set of spanning trees for $A$} is a set $S\subseteq A$ such that: (i) for each cluster $\alpha$, the subset $S[\alpha]$ of $S$ induced by the con-edges for $\alpha$ is a tree that spans the vertices belonging to $\alpha$; and (ii) there exist no two edges in $S$ that have a conflict. See Fig.~\ref{fig:con-edges}(d). The {\sc pssttm} problem asks whether a planar set of spanning trees for $A$ exists.

The following lemma relates the $c$-planarity problem for embedded flat clustered graphs to the {\sc pssttm} problem.

\begin{lemma}[\cite{df-ectefcgsf-09}] \label{le:acyclic-saturator}
An embedded flat clustered graph $C(G,T)$ is $c$-planar if and only if: (1) $G$ is planar; (2) there exists a face $f$ in $G$ such that when $f$ is chosen as outer face for $G$ no cycle composed of vertices of the same cluster encloses a vertex of a different cluster; and (3) a planar set of spanning trees for $A$ exists.
\end{lemma}

We now introduce the concept of {\em conflict graph} $K_A$, which is defined as follows. Graph $K_A$ has a vertex for each con-edge in $A$ and has an edge $(e,e')$ if $e\conf e'$. In the remainder of the paper we will show how to decide whether a set of planar spanning trees for $A$ exists by assuming that the following property holds for $A$.

\begin{prop} \label{pr:no-two-edges-same-structure}
No two con-edges for the same cluster belong to the same connected component of $K_A$.
\end{prop}

We now show that $A$ can be assumed w.l.o.g. to satisfy Property~\ref{pr:no-two-edges-same-structure}, given that $C(G,T)$ has at most two vertices per cluster incident to each face of $G$. Consider any face $f$ of $G$ and any cluster $\varrho$ such that two vertices $u_{\varrho}$ and $v_{\varrho}$ of $\varrho$ are incident to $f$.

First, no con-edge for $\varrho$ in $A$ that connects a pair of vertices different from $(u_{\varrho},v_{\varrho})$ belongs to the connected component of $K_A$ containing $(u_{\varrho},v_{\varrho})$, given that no vertex of $\varrho$ different from $u_{\varrho}$ and $v_{\varrho}$ is incident to $f$. However, it might be the case that several con-edges $(u_{\varrho},v_{\varrho})$ belong to the same connected component of $K_A$, which happens if $u_{\varrho}$, or $v_{\varrho}$, or both have several occurrences on the boundary of $f$. We show a simple reduction that gets rid of these multiple con-edges.

Denote by $B_f=(o_1,o_2,\dots,o_k)$ the clockwise order of the occurrences of vertices along the boundary of $f$ and assume w.l.o.g. that $o_i$, $o_j$, and $o_\ell$ are occurrences of  $u_{\varrho}$, $u_{\varrho}$, and $v_{\varrho}$, respectively, with $1\leq i< j< \ell\leq k$.

Suppose that there exist occurrences $o_p$ and $o_q$ in $B_f$ of vertices $x$ and $y$ belonging to a cluster $\tau$ with $\tau\neq \varrho$, with $i<p<j$, and with $j<q<\ell$, as in Fig.~\ref{fig:property1}(a). We claim that, if any planar set $S$ of spanning trees for $A$ exists, then $S$ does not contain the con-edge $e_{\varrho}=(u_{\varrho},v_{\varrho})$ connecting the occurrence $o_j$ of $u_{\varrho}$ and the occurrence $o_\ell$ of $v_{\varrho}$. Namely, all the con-edges $(x,y)$ have a conflict with $e_{\varrho}$; moreover, the con-edges $(x,y)$ form a separating set for $A[\tau]$, hence at least one of them belongs to $S$. Thus, $e_{\varrho}\notin S$, and this edge can be removed from $A$, as in Fig.~\ref{fig:property1}(b). Similar reductions can be performed if $\ell<q\leq k$ or $1\leq q<i$, and by exchanging the roles of $u_{\varrho}$ and $v_{\varrho}$.

\begin{figure}[tb]
\begin{center}
\begin{tabular}{c c c c}
\mbox{\includegraphics[scale=0.43]{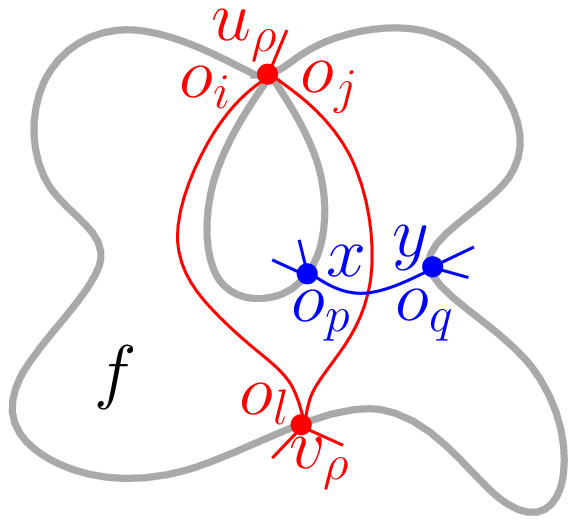}} \hspace{2mm} &
\mbox{\includegraphics[scale=0.43]{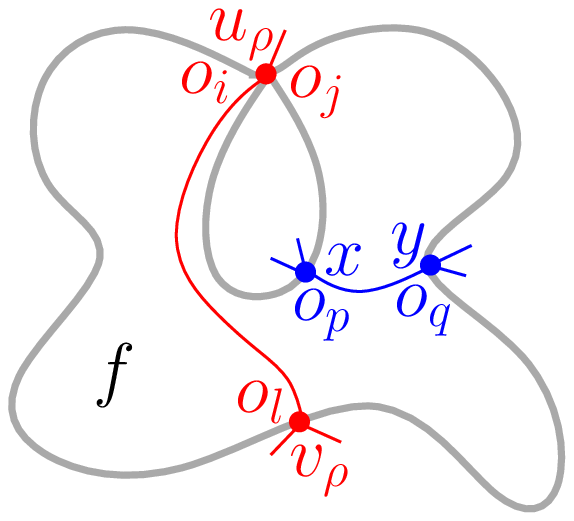}}  \hspace{2mm} &
\mbox{\includegraphics[scale=0.43]{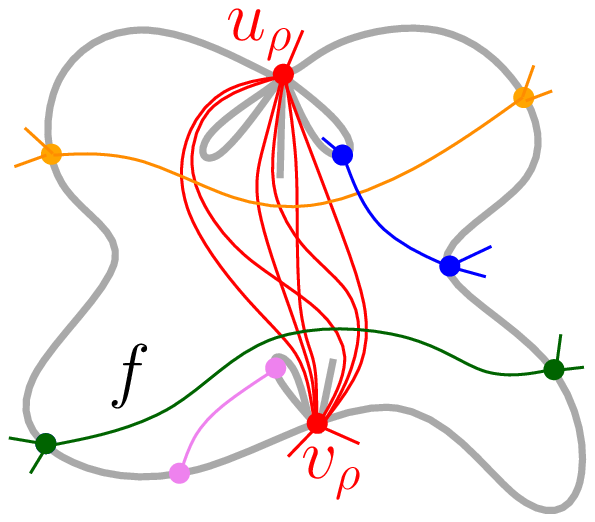}}  \hspace{2mm} &
\mbox{\includegraphics[scale=0.43]{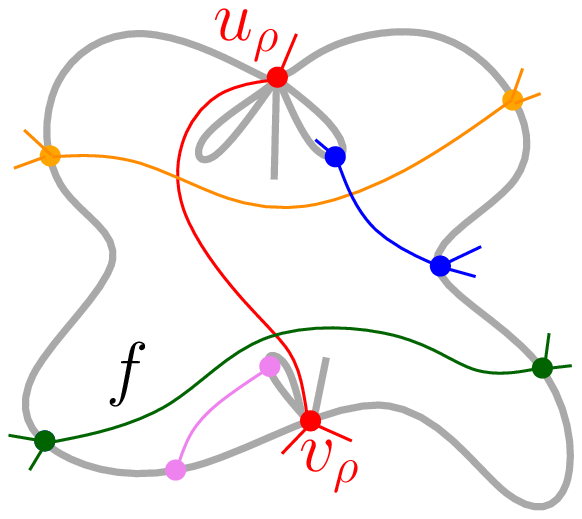}}\\
(a) \hspace{2mm} & (b) \hspace{2mm} & (c) \hspace{2mm} & (d)
\end{tabular}
\caption{Illustration for the reduction to a multigraph of the con-edges satisfying Property~\ref{pr:no-two-edges-same-structure}.}
\label{fig:property1}
\end{center}
\end{figure}

Rename the vertex occurrences in $B_f$ so that $o_1$ and $o_a$ are the first and the last occurrence of $u_{\varrho}$ in $B_f$, and so that $o_b$ and $o_c$ are the first and the last occurrence of $v_{\varrho}$ in $B_f$, with $1\leq a < b\leq c < k$. If no two occurrences $o_p$ and $o_q$ in $B_f$ as described above exist, the only con-edges $(u_{\varrho},v_{\varrho})$ left are crossed by con-edges connecting occurrences $o_p$ and $o_q$ of vertices $x$ and $y$ in $\tau$, respectively, such that $a<p<b$ and $c<q\leq k$. That is, any two con-edges $(u_{\varrho},v_{\varrho})$ cross the same set of con-edges for clusters different from $\varrho$ (see Fig.~\ref{fig:property1}(c)). Hence, a single edge $(u_{\varrho},v_{\varrho})$ can be kept in $A$, and all the other con-edges $(u_{\varrho},v_{\varrho})$ can be removed from $A$ (see Fig.~\ref{fig:property1}(d)).

After repeating this reduction for all the con-edges in $A$, an equivalent instance is eventually obtained in which Property~\ref{pr:no-two-edges-same-structure} is satisfied by $A$. Observe that the described simplification can be easily performed in $O(|C|^2)$ time. Thus, we get the following:

\begin{lemma} \label{le:pssttm}
Assume that the {\sc pssttm} problem can be solved in $f(|A|)$ time for instances satisfying Property~\ref{pr:no-two-edges-same-structure}. Then the $c$-planarity of any embedded flat clustered graph $C$ with at most two vertices per cluster on each face can be tested in $O(f(|A|)+|C|^2)$~time.
\end{lemma}

\begin{proof}
Consider any embedded flat clustered graph $C$ with at most two vertices per cluster on each face. Conditions (1) and (2) in Lemma~\ref{le:acyclic-saturator} can be tested in $O(|C|)$ time (see~\cite{df-ectefcgsf-09}); hence, testing the $c$-planarity of $C$ is equivalent to solve the {\sc pssttm} problem for $A$. Finally, as described before the lemma, there exists an $O(|C|^2)$-time algorithm that modifies multigraph $A$ so that it satisfies Property~\ref{pr:no-two-edges-same-structure}.
\end{proof}

Before proceeding with the description of the algorithm, we state a direct consequence of Property~\ref{pr:no-two-edges-same-structure} that will be useful in the upcoming proofs. Refer to Fig.~\ref{fig:inside-a-face}. Consider a set $F\subseteq A$ of con-edges all belonging to the same connected component of $K_A$ and such that all their end-vertices are incident to the outer face of the subgraph of $A$ induced by $F$. Let $\Delta_F$ be the set of clusters that have con-edges in $F$. Then, it is possible to draw a closed curve $\cal C$ that passes through the end-vertices of all the edges in $F$, that contains all the con-edges in $F$ in its interior, and all the other con-edges for clusters in $\Delta_F$ in its exterior.

\begin{figure}[htb]
\begin{center}
\mbox{\includegraphics[scale=0.5]{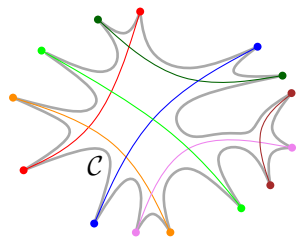}}
\caption{Curve $\cal C$ (gray) and edges in $F$ (multi-colored).}
\label{fig:inside-a-face}
\end{center}
\end{figure}

\section{Algorithm Outline} \label{se:outline}

In this section we give an outline of our algorithm for testing the existence of a planar set $S$ of spanning trees for $A$, where we assume that no two con-edges for the same cluster belong to the same connected component of $K_A$.

Our algorithm repeatedly tries to detect certain substructures in~$A$. When it does find one of such substructures, the algorithm either ``simplifies'' $A$ or concludes that $A$ does not admit any planar set of spanning trees. For example, if a cluster $\alpha$ exists such that $A[\alpha]$ is not connected, then the algorithm concludes that no planar set of spanning trees exists and terminates; as another example, if conflicting con-edges $e_{\alpha}$ and $e_{\beta}$ for clusters $\alpha$ and $\beta$ exist in $A$ such that $e_{\alpha}$ is a bridge for $A[\alpha]$, then the algorithm determines that $e_{\alpha}$ has to be in $S$ and that $e_{\beta}$ can be assumed not to be in $S$.

If the algorithm determines that certain edges have to be in $S$ or can be assumed not to be in $S$, these edges are contracted or removed, respectively. Given a set $A'\subseteq A$, the operation of {\em removing} $A'$ from $A$ consists of updating $A:=A\setminus A'$. Given a set $A'\subseteq A$, the operation of {\em contracting} the edges in $A'$ consists of identifying the end-vertices of each con-edge $e$ in $A'$ (all the con-edges different from $e$ and incident to the end-vertices of $e$ remain in $A$), and of updating $A:=A\setminus A'$.

The edges in $A'$ are removed from $A$ (contracted in $A$) only when this operation does not alter the possibility of finding a planar set of spanning trees for $A$. Also, contractions are only applied to con-edges that cross no other con-edges; hence after any contraction graph $K_A$ only changes because of the removal of the isolated vertices corresponding to the contracted edges.

As a consequence of a removal or of a contraction operation, the number of edges in $A$ decreases, that is, $A$ is ``simplified''. After any simplification due to the detection of a certain substructure in $A$, the algorithm will run again all previous tests for the detection of the other substructures. In fact, it is possible that a certain substructure arises from performing a simplification on $A$ (e.g., a bridge might be present in $A$ after a set of edges has been removed from $A$). Since detecting each substructure that leads to a simplification in $A$ can be performed in quadratic time, and since the initial size of $A$ is linear in the size of $C$, the algorithm has a cubic running time.

If none of the four tests (called {\sc Test 1--4}) and none of the eight simplifications (called {\sc Simplification 1--8}), that will be fully described in Section~\ref{se:algorithm}, applies to $A$, then $A$ is a {\em single-conflict} multigraph. That is, each con-edge in $A$ crosses at most one con-edge in $A$. A linear-time algorithm for deciding the existence of a planar set of spanning trees in a single-conflict multigraph $A$ is known~\cite{df-ectefcgsf-09}. Hence, our algorithm uses that algorithm~\cite{df-ectefcgsf-09} to conclude the test of the existence of a planar set of spanning trees in $A$. A pseudo-code description of our algorithm is presented in Algorithm~\ref{a:alg}.

\newcommand{\lIf}[2]{\State\algorithmicif\ {#1}\ \algorithmicthen\ {#2}}
\newcommand{\lElse}[1]{\State\algorithmicelse\ {#1}}
\renewcommand{\algorithmiccomment}[1]{\hfill{\color[rgb]{0,0.5,0}$\triangleright$ #1}}
\begin{algorithm}
\begin{algorithmic}[1]
\State $S=\emptyset$;
 \While{$\exists$ con-edge that crosses more than one con-edge in $A$}\label{step:restart}
     \If{$\exists$ cluster $\alpha$ such that $A[\alpha]$ is disconnected}
        \State \Return ``$\nexists$ planar set of spanning trees for $A$'' \Comment{{\sc Test 1} (L\ref{le:disconnected})}
     \EndIf
     \If{$\exists$ $e$ that is a bridge of $A[\alpha]$}
        \State Remove $L_1(e)$ from $A$, insert $e$ in $S$, contract $e$ in $A$, goto \eqref{step:restart} \Comment{{\sc Simpl.~1} (L\ref{le:bridge})}
     \EndIf
     \If{$\exists$ $k\geq 1,  e_1, \dots, e_{2k+1}, e_{2k+2}=e_1\in A$, s.t. $e_i\conf e_{i+1}$, for $1\leq i\leq 2k +1$}
        \State \Return ``$\nexists$ planar set of spanning trees for $A$'' \Comment{{\sc Test 2} (L\ref{le:conflicts-are-bipartite})}
     \EndIf
     \If{$\exists$ con-edge $e$ that is a self-loop}
        \State Remove $e$ from $A$, goto \eqref{step:restart} \Comment{{\sc Simpl. 2} (L\ref{le:self-loop})}
     \EndIf
     \If{$\exists$ con-edge $e\in A$ that does not cross any con-edge in $A$}
       \State Insert $e$ in $S$, contract $e$ in $A$, goto \eqref{step:restart} \Comment{{\sc Simpl. 3} (L\ref{le:non-conf})}
     \EndIf
     \If{$\exists$ $e_{\alpha} \in A[\alpha]$, $e_{\beta} \in A[\beta]$, $e_{\gamma} \in A[\gamma]$ with $e_{\alpha} \conf e_{\beta}, e_{\gamma}$, $\exists$ facial cycle ${\cal C}_{\alpha}$ of $A[\alpha]$, and \\\hspace{0.38cm} $\nexists$ $e'_{\alpha} \in A[\alpha]$,  $e'_{\beta} \in A[\beta]$, $e'_{\gamma} \in A[\gamma]$ with $e'_{\alpha} \conf e'_{\beta}, e'_{\gamma}$ s.t. $e_{\alpha},e'_{\alpha}\in {\cal C}_{\alpha}$}
        \State Remove $e_{\alpha}$ from $A$, goto \eqref{step:restart} \Comment{{\sc Simpl. 4} (L\ref{le:deviation})}
     \EndIf
     \If{$\exists$ $e_{\alpha},e'_{\alpha}\in A[\alpha]$ sharing a face of $A[\alpha]$ delimited by cycle ${\cal C}_{\alpha}$ and both crossed first \\\hspace{0.38cm} by a con-edge for $\beta$ and then by a con-edge for $\gamma$ when traversing ${\cal C}_{\alpha}$ clockwise}
        \State \Return ``$\nexists$ planar set of spanning trees for $A$'' \Comment{{\sc Test 3} (L\ref{le:same-order})}
     \EndIf
     \If{$\exists$ $e_{\alpha},e'_{\alpha}, e''_{\alpha}$ in facial cycle ${\cal C}_{\alpha}$ of $A[\alpha]$, s.t. $e_{\alpha}$, $e''_{\alpha}$, $e'_{\alpha}$ are encountered in this  order \\\hspace{0.38cm} when traversing ${\cal C}_{\alpha}$ clockwise, $\exists$ $e_{\beta}, e'_{\beta}, e''_{\beta} \in A[\beta]$, $\exists$ $e_{\gamma}, e'_{\gamma}\in A[\gamma]$, s.t. $e_{\alpha} \conf e_{\beta},e_{\gamma}$, \\\hspace{0.38cm} $e'_{\alpha} \conf e'_{\beta},e'_{\gamma}$, and $e''_{\alpha} \conf e''_{\beta}$, and such that $e_{\alpha}$ is crossed first by $e_{\beta}$ and then by $e_{\gamma}$ when \\\hspace{0.38cm} traversing ${\cal C}_{\alpha}$ clockwise}
     \State \Return ``$\nexists$ planar set of spanning trees for $A$'' \Comment{{\sc Test 4} (L\ref{le:ababa})}
     \EndIf
     \If{$\exists$ $\alpha$-donut with spokes $e^i_{\alpha}$ and $e^{i+1}_{\alpha}$ s.t. $M(e^i_{\alpha})$ is isomorphic to $M(e^{i+1}_{\alpha})$}
          \State Remove $\bigcup_j H_j(e^i_{\alpha})$ from $A$, insert  $\bigcup_j L_j(e^i_{\alpha})$ in $S$, contract $\bigcup_j L_j(e^i_{\alpha})$ in $A$, \\\hspace{0.9cm} goto \eqref{step:restart} \Comment{{\sc Simpl. 5} (L\ref{le:isomorphic})}
     \EndIf
     \If{$\exists$ $\alpha$-donut $D_{\alpha}$ with spokes $e^i_{\alpha}$ and $e^{i+1}_{\alpha}$, $\exists$ $e^i_{\beta}\in L_1(e^i_{\alpha})\cap A[\beta]$, $e^{i+1}_{\beta}\in L_1(e^{i+1}_{\alpha}) \cap$ \\\hspace{0.38cm} $A[\beta]$, $\exists$ $e^i_{\gamma}\in H_1(e^i_{\alpha})\cap A[\gamma]$ s.t. $e^i_{\gamma}\conf e^i_{\beta}$, and $e^i_{\gamma}$ is in $f^{i+1}_{\alpha}$}
         \If{$\nexists$ con-edge $e^{i+1}_{\gamma}\in H_1(e^{i+1}_{\alpha})$ for $\gamma$ s.t. $e^{i+1}_{\gamma}\conf e^{i+1}_{\beta}$}
            \State Remove $\bigcup_j L_j(e^i_{\alpha})$ from $A$, insert $\bigcup_j H_j(e^i_{\alpha})$ in $S$, contract $\bigcup_j H_j(e^i_{\alpha})$ in $A$, \\\hspace{1.35cm} goto \eqref{step:restart} \Comment{{\sc Simpl. 6} (L\ref{le:non-isomorphic-differentT1})}
         \EndIf
         \If{$\exists$ $e^{i+1}_{\gamma}$$\in$$H_1(e^{i+1}_{\alpha})\cap A[\gamma]$ s.t. $e^{i+1}_{\gamma}\conf e^{i+1}_{\beta}$ and $\exists$ spoke $e^{i+2}_{\alpha}$$\neq$$e^i_{\alpha},e^{i+1}_{\alpha}$ of $D_{\alpha}$}
            \State Remove $\bigcup_j H_j(e^{i+2}_{\alpha})$ from $A$, insert $\bigcup_j L_j(e^{i+2}_{\alpha})$ in $S$, contract $\bigcup_j L_j(e^{i+2}_{\alpha})$ \\\hspace{1.35cm} in $A$, goto \eqref{step:restart} \Comment{{\sc Simpl. 7} (L\ref{le:non-isomorphic-sameT1})}
         \EndIf
     \EndIf
     \If{$\exists$ $\alpha$-donut with exactly two spokes $e^1_{\alpha}$ and $e^2_{\alpha}$, $\exists$  $j\geq 1$ s.t. (1) $\exists$ $e_{\mu}\in L_j(e^a_{\alpha})\cap A[\mu]$ \\\hspace{0.38cm} and $e_{\nu}\in H_{j-1}(e^a_{\alpha})\cap A[\nu]$ s.t. $e_{\mu}\conf e_{\nu}$, and $\nexists$  $g_{\mu}\in L_j(e^b_{\alpha})\cap A[\mu]$ s.t. $g_{\mu}\conf g_{\nu}$ with \\\hspace{0.38cm} $g_{\nu}\in$ $H_{j-1}(e^b_{\alpha})\cap A[\nu]$, or (2) $\exists$ $e_{\mu}\in H_j(e^a_{\alpha})\cap A[\mu]$ and $e_{\nu}\in L_{j}(e^a_{\alpha})\cap A[\nu]$ s.t. \\\hspace{0.38cm} $e_{\mu}\conf e_{\nu}$, and $\nexists$ $g_{\mu}\in H_j(e^b_{\alpha})\cap A[\mu]$ s.t. $g_{\mu}\conf g_{\nu}$ with $g_{\nu}\in L_{j}(e^b_{\alpha})\cap A[\nu]$}
         \State Let $j$ be the minimal integer satisfying (1) or (2)\Comment{{\sc Simpl. 8} (L\ref{le:non-isomorphic-differentTj})}
         \lIf{$e_{\mu}\in L_j(e^a_{\alpha})$} remove $\bigcup_j H_j(e^a_{\alpha})$ from $A$, insert $\bigcup_j L_j(e^a_{\alpha})$ in $S$, \\\hspace{3.55cm} contract $\bigcup_j L_j(e^a_{\alpha})$ in $A$, goto \eqref{step:restart}
         \lIf{$e_{\mu}\in H_j(e^a_{\alpha})$} remove $\bigcup_j L_j(e^a_{\alpha})$ from $A$, insert $\bigcup_j H_j(e^a_{\alpha})$ in $S$, \\\hspace{3.62cm} contract $\bigcup_j H_j(e^a_{\alpha})$ in $A$, goto \eqref{step:restart}
     \EndIf
\EndWhile
\State \Return the output of the algorithm in~\cite{df-ectefcgsf-09} on $A$ \Comment{(L\ref{le:no-simplification-single-conflict})}
\end{algorithmic}
\caption{Testing for the existence of a planar set $S$ of spanning trees for $A$\label{a:alg}. The comments specify each test and simplification, and the lemma proving its correctness.}
\end{algorithm}


\section{Algorithm} \label{se:algorithm}

To ease the reading and avoid text duplication, when introducing a new lemma we always assume, without making it explicit, that all the previously defined simplifications do not apply, and that all the previously defined tests fail. Also, we do not make explicit the removal and contraction operations that we perform, as they straight-forwardly follow from the statement of each lemma. Refer also to the description in Algorithm~\ref{a:alg}.

We start with the following test.

\begin{lemma}[{\sc Test 1}]\label{le:disconnected}
Let $\alpha$ be a cluster such that $A[\alpha]$ is disconnected. Then, there exists no planar set $S$ of spanning trees for $A$.
\end{lemma}

\begin{proof}
No set $S\subseteq A$ is such that $S[\alpha]$ induces a graph that spans the vertices belonging to $\alpha$. This proves the lemma.
\end{proof}


If a con-edge $e$ is a bridge for some graph $A[\alpha]$, then not choosing $e$ to be in $S$ would disconnect $A[\alpha]$, which implies that no planar set of spanning trees for $A$ exists.

\begin{lemma}[{\sc Simplification 1}]\label{le:bridge}
Let $e$ be a bridge of $A[\alpha]$. Then, for every planar set $S$ of spanning trees for $A$, we have $e\in S$.
\end{lemma}

\begin{proof}
Suppose, for a contradiction, that a planar set $S$ of spanning trees for $A$ exists such that $e\notin S$. Then $A[\alpha]\setminus\{e\}$ is disconnected. By Lemma~\ref{le:disconnected}, no planar set of spanning trees for $A\setminus\{e\}$ exists, a contradiction.
\end{proof}


The following lemma is used massively in the remainder of the paper.

\begin{lemma} \label{le:one-or-the-other}
Let $e_{\alpha},e_{\beta}\in A$ be con-edges such that $e_{\alpha}\conf e_{\beta}$. Let $S$ be a planar set of spanning trees for $A$ and suppose that $e_{\alpha} \notin S$. Then, $e_{\beta} \in S$.
\end{lemma}

\begin{proof}
Assume, for a contradiction, that $S$ contains neither $e_{\alpha}$ nor $e_{\beta}$. Then, there exists a path $P_{\alpha}\not\ni e_{\alpha}$ ($P_{\beta}\not\ni e_{\beta}$) all of whose edges belong to $S$ connecting the end-vertices of $e_{\alpha}$ (resp.\ of $e_{\beta}$). Consider the cycle ${\cal C}_{\alpha}$ composed of $P_{\alpha}$ and $e_{\alpha}$. We have that $P_{\beta}$ cannot cross  ${\cal C}_{\alpha}$. In fact, $P_{\beta}$ cannot cross $P_{\alpha}$, as both such paths are composed of con-edges in $S$, and it cannot cross $e_{\alpha}$ by Property~\ref{pr:no-two-edges-same-structure}, given that $e_{\alpha}\conf e_{\beta}$ and $e_{\beta}\notin P_{\beta}$. However, the end-vertices of $e_{\beta}$ are on different sides of ${\cal C}_{\alpha}$, hence by the Jordan curve theorem $P_{\beta}$ does cross ${\cal C}_{\alpha}$, a contradiction.
\end{proof}


The algorithm continues with the following test.


\begin{lemma}[{\sc Test 2}]\label{le:conflicts-are-bipartite}
If the conflict graph $K_A$ is not bipartite, then there exists no planar set $S$ of spanning trees for~$A$.
\end{lemma}

\begin{proof}
Assume, for a contradiction, that $K_A$ is not bipartite and that $S$ exists. Let $e_1, \dots, e_{2k+1},e_{2k+2}=e_1$ be a cycle in $K_A$ with an odd number of vertices (recall that vertices in $K_A$ are con-edges in $A$). Suppose that $e_1\in S$. Then, by repeated applications of Lemma~\ref{le:one-or-the-other} and of the fact that $S$ does not contain two conflicting edges, we get $e_2\notin S$, $e_3\in S$, $e_4\notin S$, $e_5\in S$, $\dots$, $e_{2k+1}\in S$, $e_{2k+2}=e_1\notin S$, a contradiction. The case in which $e_1\notin S$ can be discussed analogously.
\end{proof}



The contraction of con-edges in $A$ that have been chosen to be in $S$ might lead to self-loops in $A$, a situation that is dealt with in the following.
\begin{lemma}[{\sc Simplification 2}]\label{le:self-loop}
Let $e\in A$ be a self-loop. Then, for every planar set $S$ of spanning trees for $A$, we have $e\notin S$.
\end{lemma}

\begin{proof}
Since a tree does not contain any self-loop, the lemma follows.
\end{proof}


Next, we show a simplification that can be performed if a con-edge exists in $A$ that does not have a conflict with any other con-edge in $A$.
\begin{lemma}[{\sc Simplification 3}]\label{le:non-conf}
Let $e$ be any con-edge in~$A$ that does not have a conflict with any other con-edge in $A$. Then, there exists a planar set $S$ of spanning trees for $A$ if and only if there exists a planar set $S'$ of spanning trees for $A$ such that $e\in S'$.
\end{lemma}

\begin{proof}
Let $S$ be any planar set of spanning trees for $A$. If $e\in S$, then there is nothing to prove. Suppose that $e\notin S$. Since $e$ does not cross any con-edge in $A$, we have that $S\cup \{e\}$ does not contain any two conflicting edges. Denote by $\alpha$ the cluster $e$ is a con-edge for. Since $S[\alpha]$ is a spanning tree, $S[\alpha]\cup \{e\}$ contains a cycle $\cal C$. Since we can assume that {\sc Simplification 2} does not apply to $A$ (it would have been performed before applying this lemma), we have that $\cal C$ contains at least one edge $e'$ different from $e$. Then, $S\cup\{e\}\setminus\{e'\}$ is a planar set of spanning trees for $A$. 
\end{proof}

In the next three lemmata we deal with the following setting. Assume that there exist con-edges $e_{\alpha}, e_{\beta}, e_{\gamma}\in A$ for distinct clusters $\alpha$, $\beta$, and $\gamma$, respectively, such that $e_{\alpha}\conf e_{\beta}$ and $e_{\alpha} \conf e_{\gamma}$. Since {\sc Test 2} fails on $A$, $e_{\beta}$ does not cross $e_{\gamma}$. Let ${\cal C}_{\alpha}$ be any of the two facial cycles of $A[\alpha]$ incident to $e_{\alpha}$, where a facial cycle of  $A[\alpha]$ is a simple cycle all of whose edges appear on the boundary of a single face of $A[\alpha]$.  Assume w.l.o.g.\ that $e_{\alpha}$ is crossed first by $e_{\beta}$ and then by $e_{\gamma}$ when ${\cal C}_{\alpha}$ is traversed clockwise.  See Fig.~\ref{fig:many-conflicts}(a).

\begin{figure}[tb]
\begin{center}
\begin{tabular}{c c c}
\mbox{\includegraphics[scale=0.39]{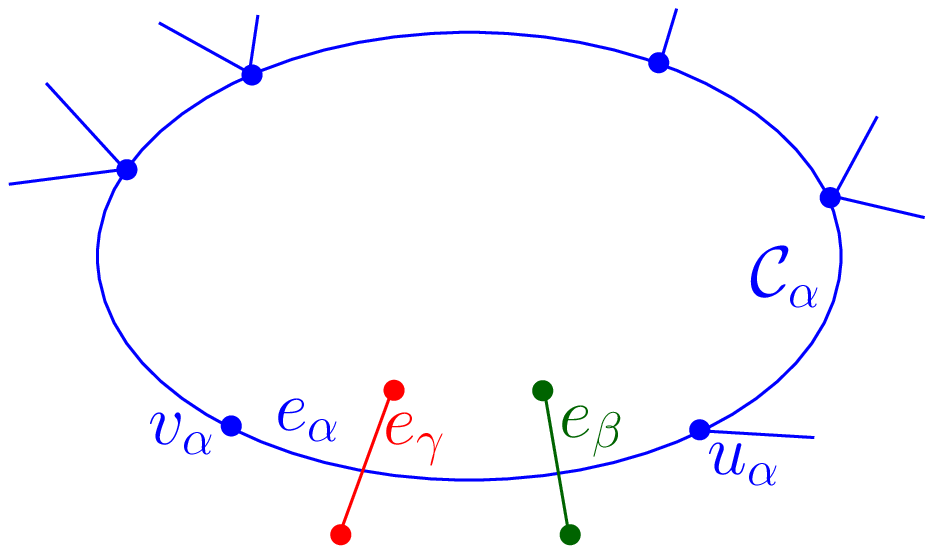}} \hspace{1mm} &
\mbox{\includegraphics[scale=0.39]{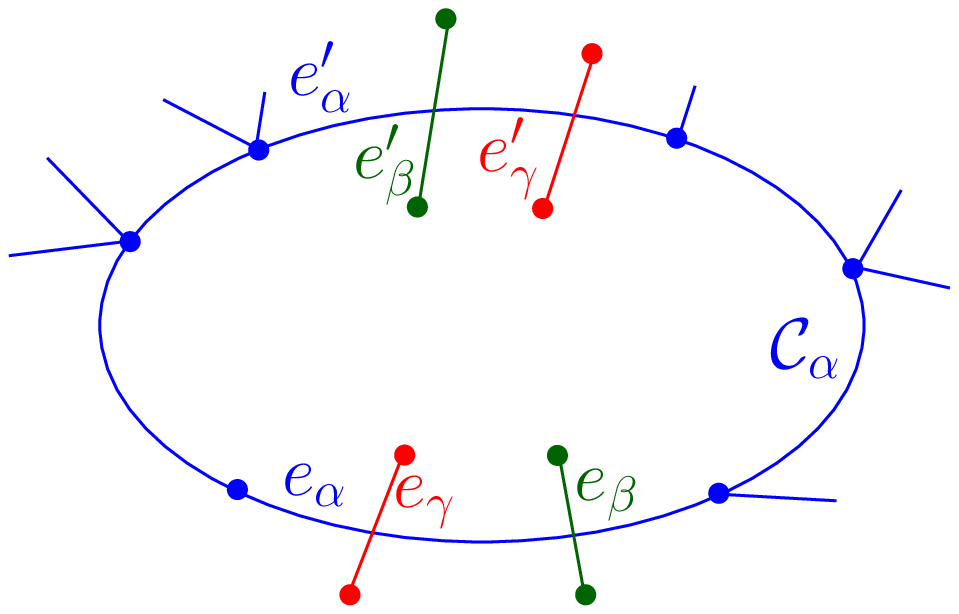}} \hspace{1mm} &
\mbox{\includegraphics[scale=0.39]{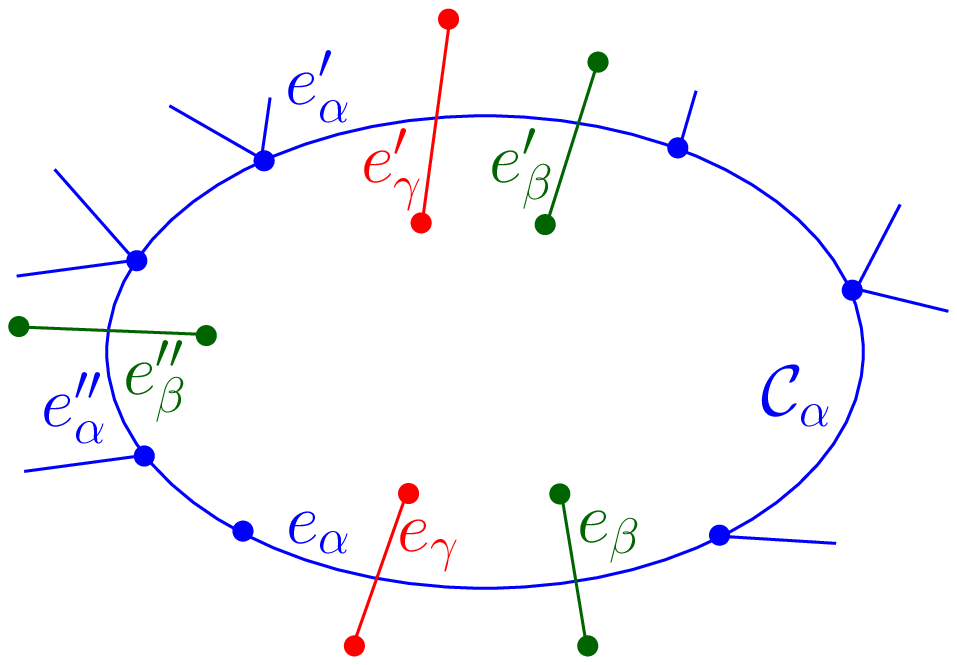}}\\
(a) \hspace{1mm} & (b) \hspace{1mm} & (c)
\end{tabular}
\caption{The setting for (a) Lemma~\ref{le:deviation}, (b) Lemma~\ref{le:same-order}, and (c) Lemma~\ref{le:ababa}.}
\label{fig:many-conflicts}
\end{center}
\end{figure}



The next lemma presents a condition in which we can delete an edge $e_a$ from $A$.

\begin{lemma}[{\sc Simplification 4}]\label{le:deviation}
Suppose that there exists no con-edge of ${\cal C}_{\alpha}$ different from $e_{\alpha}$ that has a conflict with both a con-edge for $\beta$ and a con-edge for $\gamma$. Then, for every planar set $S$ of spanning trees for $A$, we have $e_{\alpha}\notin S$.
\end{lemma}

\begin{proof}
Denote by $u_{\alpha}$ and $v_{\alpha}$ (by $u_{\beta}$ and $v_{\beta}$, by $u_{\gamma}$ and $v_{\gamma}$) the end-vertices of $e_{\alpha}$ (resp.\ of $e_{\beta}$, resp.\ of $e_{\gamma}$). By Property~\ref{pr:no-two-edges-same-structure}, it is possible to draw a closed curve $\cal C$ passing through $u_{\alpha}$, $u_{\beta}$, $u_{\gamma}$, $v_{\alpha}$, $v_{\gamma}$, and $v_{\beta}$ in (w.l.o.g.) clockwise order, containing edges $e_{\alpha}$, $e_{\beta}$, and $e_{\gamma}$ in its interior, and containing every other con-edge for $\alpha$, $\beta$, and $\gamma$ in its exterior. See Fig.~\ref{fig:lemma-deviation}.

Suppose, for a contradiction, that there exists a planar set $S$ of spanning trees for $A$ such that $e_{\alpha}\in S$.
\begin{figure}[tb]
\begin{center}
\mbox{\includegraphics[scale=0.4]{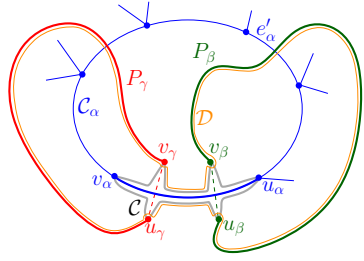}}
\caption{Illustration for the proof of Lemma~\ref{le:deviation}.}
\label{fig:lemma-deviation}
\end{center}
\end{figure}
Then, there exists a path $P_{\beta}$ ($P_{\gamma}$) all of whose edges belong to $S$ connecting $u_{\beta}$ and $v_{\beta}$ (resp.\ $u_{\gamma}$ and $v_{\gamma}$). Since $u_{\beta}$ and $v_{\beta}$ are on different sides of ${\cal C}_{\alpha}$, by the Jordan curve theorem $P_{\beta}$ crosses a con-edge $e'_{\alpha}\neq e_{\alpha}$ of ${\cal C}_{\alpha}$. Since no con-edge of ${\cal C}_{\alpha}$ different from $e_{\alpha}$ has a conflict with both a con-edge for $\beta$ and a con-edge for $\gamma$, it follows that $e'_{\alpha}$ does not cross $P_{\gamma}$. Consider the cycle ${\cal D}$ composed of $P_{\beta}$, of $P_{\gamma}$, of the path $P_u$ in $\cal C$ between $u_{\beta}$ and $u_{\gamma}$ not containing $u_{\alpha}$, and of the path $P_v$ in $\cal C$ between $v_{\beta}$ and $v_{\gamma}$ not containing $v_{\alpha}$. There exist vertices of $\alpha$ on both sides of cycle ${\cal D}$ (e.g., the end-vertices of $e'_{\alpha}$). However, no con-edge $g_{\alpha}$ for $\alpha$ in $S$ can cross ${\cal D}$. In fact, $g_{\alpha}$ cannot cross $P_{\beta}$ and $P_{\gamma}$, as such paths are composed of con-edges in $S$, and it cannot cross $P_u$ and $P_v$ by construction of $\cal C$. It follows that $S$ does not connect $\alpha$, a contradiction to the fact that $S$ is a planar set of spanning trees for $A$.
\end{proof}



The next two lemmata state conditions in which no planar set of spanning trees for $A$ exists. Their statements are illustrated in Figs.~\ref{fig:many-conflicts}(b) and~\ref{fig:many-conflicts}(c), respectively.

\begin{lemma}[{\sc Test 3}]\label{le:same-order}
Suppose that there exist con-edges $e'_{\alpha}, e'_{\beta}, e'_{\gamma}\in A$ for clusters $\alpha$, $\beta$, and $\gamma$, respectively, such that  $e'_{\alpha}\neq e_{\alpha}$, $e'_{\alpha}$ belongs to ${\cal C}_{\alpha}$, and $e'_{\alpha}\conf e'_{\beta}$ as well as $e'_{\alpha}\conf e'_{\gamma}$. Assume that $e'_{\alpha}$ is crossed first by $e'_{\beta}$ and then by $e'_{\gamma}$ when ${\cal C}_{\alpha}$ is traversed clockwise. Then, no planar set of spanning trees for $A$ exists.
\end{lemma}

\begin{proof}
Denote by $u_{\alpha}$ and $v_{\alpha}$ (by $u_{\beta}$ and $v_{\beta}$, by $u_{\gamma}$ and $v_{\gamma}$) the end-vertices of $e_{\alpha}$ (resp.\ of $e_{\beta}$, resp.\ of $e_{\gamma}$). By Property~\ref{pr:no-two-edges-same-structure}, it is possible to draw a closed curve $\cal C$ passing through $u_{\alpha}$, $u_{\beta}$, $u_{\gamma}$, $v_{\alpha}$, $v_{\gamma}$, and $v_{\beta}$ in this clockwise order, containing edges $e_{\alpha}$, $e_{\beta}$, and $e_{\gamma}$ in its interior, and containing every other con-edge for $\alpha$, $\beta$, and $\gamma$ in its exterior.

Denote by $u'_{\alpha}$ and $v'_{\alpha}$ (by $u'_{\beta}$ and $v'_{\beta}$, by $u'_{\gamma}$ and $v'_{\gamma}$) the end-vertices of $e'_{\alpha}$ (resp.\ of $e'_{\beta}$, resp.\ of $e'_{\gamma}$). By Property~\ref{pr:no-two-edges-same-structure}, it is possible to draw a closed curve ${\cal C}'$ passing through $u'_{\alpha}$, $u'_{\beta}$, $u'_{\gamma}$, $v'_{\alpha}$, $v'_{\gamma}$, and $v'_{\beta}$ in this clockwise order, containing edges $e'_{\alpha}$, $e'_{\beta}$, and $e'_{\gamma}$ in its interior, and containing every other con-edge for $\alpha$, $\beta$, and $\gamma$ in its exterior. Assume, w.l.o.g., that the face of $A[\alpha]$ delimited by ${\cal C}_{\alpha}$ is to the right of ${\cal C}_{\alpha}$ when traversing such a cycle in clockwise direction. Finally assume, w.l.o.g., that $u_{\alpha}$, $v_{\alpha}$, $u'_{\alpha}$, and $v'_{\alpha}$ appear in this clockwise order along ${\cal C}_{\alpha}$ (possibly $v_{\alpha}=u'_{\alpha}$ and/or $v'_{\alpha}=u_{\alpha}$).

Suppose, for a contradiction, that there exists a planar set $S$ of spanning trees for $A$. The proof distinguishes two cases.

\paragraph{Case 1.} Let $e_{\alpha},e'_{\alpha} \notin S$. Refer to Fig.~\ref{fig:same-order}(a).

Consider the path $P^u_{\alpha}$ ($P^v_{\alpha}$) all of whose edges belong to $S$ connecting $u_{\alpha}$ and $u'_{\alpha}$ (resp.\ $v_{\alpha}$ and $v'_{\alpha}$). Since the end-vertices of $P^u_{\alpha}$ and $P^v_{\alpha}$ alternate along ${\cal C}_{\alpha}$ and since ${\cal C}_{\alpha}$ delimits a face of $A[\alpha]$, it follows that $P^u_{\alpha}$ and $P^v_{\alpha}$ share vertices (and possibly edges). Thus, the union of $P^u_{\alpha}$ and $P^v_{\alpha}$ is a tree $T_{\alpha}$ (recall that $S[\alpha]$ contains no cycle) whose only possible leaves are $u_{\alpha}$, $v_{\alpha}$, $u'_{\alpha}$, and $v'_{\alpha}$. Next, consider the path $P^u_{\beta}$ all of whose edges belong to $S$ connecting $u_{\beta}$ and $u'_{\beta}$. We claim that this path contains $v_{\beta}$. Indeed, if $P^u_{\beta}$ does not contain $v_{\beta}$, then it crosses the path connecting $u_{\alpha}$ and $v_{\alpha}$ in $T_{\alpha}$, thus contradicting the fact that $S$ is a planar set of spanning trees. An analogous proof shows that $P^u_{\beta}$ contains~$v'_{\beta}$.

Now, consider the cycle ${\cal D}$ composed of $P^u_{\alpha}$, of $P^u_{\beta}$, of the path $P_u$ in $\cal C$ between $u_{\alpha}$ and $u_{\beta}$ not containing $u_{\gamma}$, and of the path $P'_u$ in ${\cal C}'$ between $u'_{\alpha}$ and $u'_{\beta}$ not containing $u'_{\gamma}$. Cycle ${\cal D}$ contains vertices of $\gamma$ on both sides (e.g., $u_{\gamma}$ and $u'_{\gamma}$). However, no con-edge $g_{\gamma}$ for $\gamma$ in $S$ can cross ${\cal D}$. In fact, $g_{\gamma}$ cannot cross $P^u_{\alpha}$ or $P^u_{\beta}$, as such paths are composed of con-edges in $S$, and it cannot cross $P_u$ and $P'_u$ by construction of ${\cal C}$ and ${\cal C}'$. It follows that $S$ does not connect $\gamma$, a contradiction to the fact that $S$ is a planar set of spanning trees.

\begin{figure}[tb]
\begin{center}
\begin{tabular}{c c}
\mbox{\includegraphics[scale=0.38]{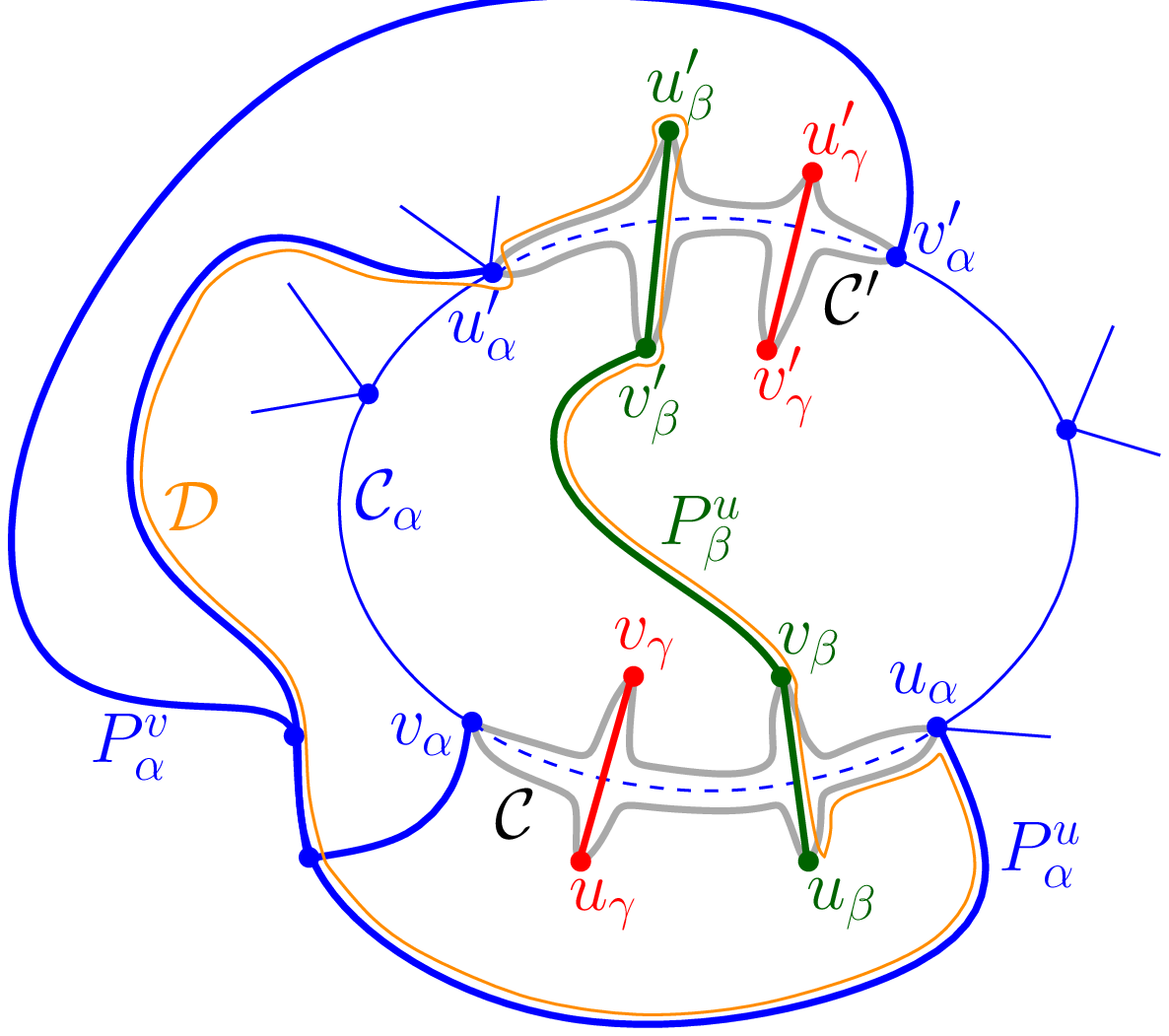}} \hspace{1mm} &
\mbox{\includegraphics[scale=0.38]{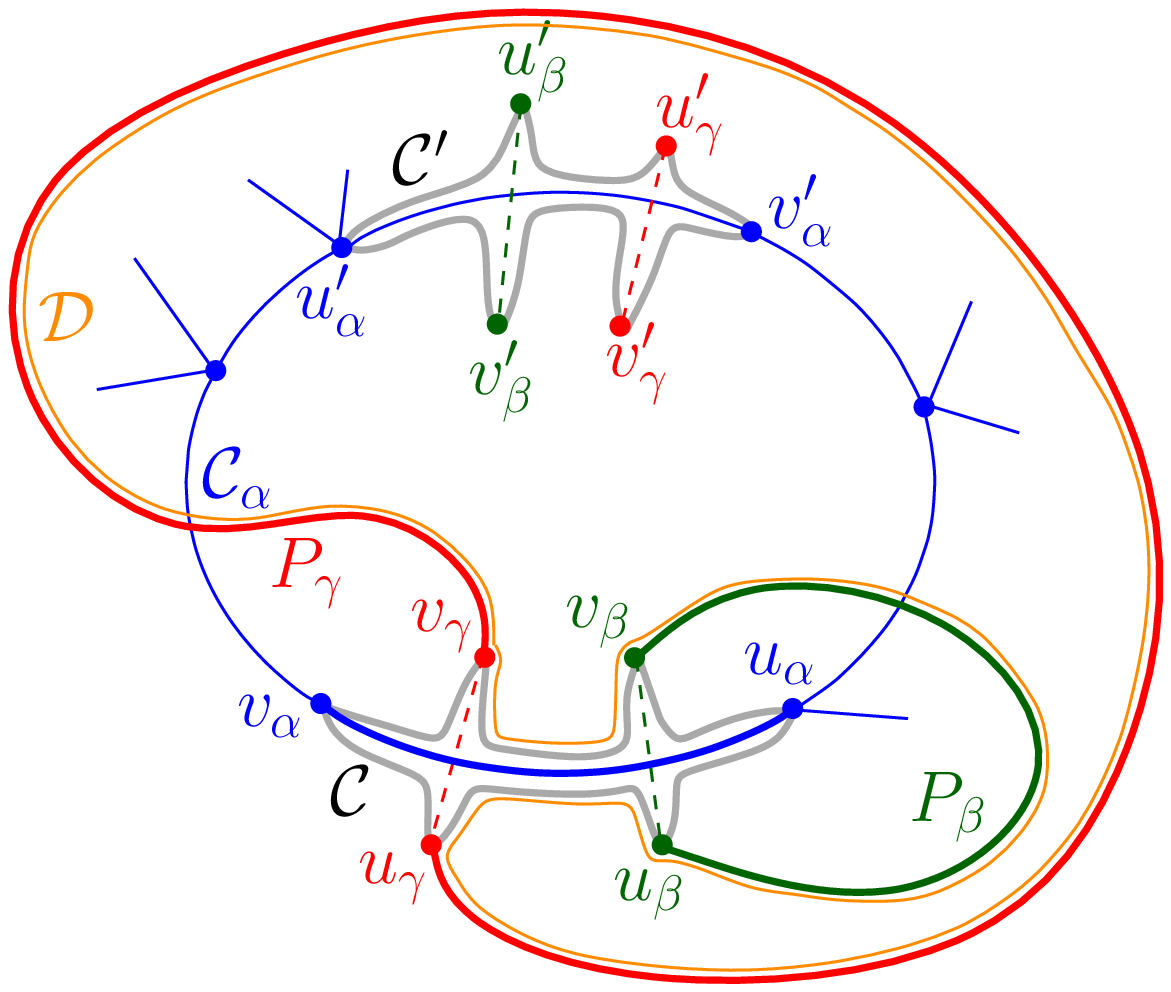}}\\
(a) \hspace{1mm} & (b)
\end{tabular}
\caption{Proof of Lemma~\ref{le:same-order}. (a) The case in which neither $e_{\alpha}$ nor $e'_{\alpha}$ belongs to $S$. (b) The case in which $e_{\alpha}\in S$.}
\label{fig:same-order}
\end{center}
\end{figure}

\paragraph{Case 2.} Let $\{e_{\alpha},e'_{\alpha}\}\cap S\neq \emptyset$. Suppose, w.l.o.g., that $e_{\alpha}\in S$. Refer to Fig.~\ref{fig:same-order}(b).

Consider the path $P_{\beta}$ ($P_{\gamma}$) all of whose edges belong to $S$ connecting $u_{\beta}$ and $v_{\beta}$ (resp.\ $u_{\gamma}$ and $v_{\gamma}$). Consider the cycle ${\cal C}_{\beta}$ composed of $P_{\beta}$ and $e_{\beta}$. We have that no con-edge $g_{\gamma}$ for $\gamma$ in $S$ can cross ${\cal C}_{\beta}$. In fact, $g_{\gamma}$ cannot cross $P_{\beta}$, as such a path is composed of con-edges in $S$, and it cannot cross $e_{\beta}$ by Property~\ref{pr:no-two-edges-same-structure}, given that $e_{\beta}$ and $e_{\gamma}$ belong to the same connected component of $K_A$ and do not cross, as otherwise {\sc Test 2} would succeed on $A$. It follows that ${\cal C}_{\beta}$ has $u_{\gamma}$ and $u'_{\gamma}$ on the same side, as otherwise $S$ would not connect $\gamma$, a contradiction to the fact that $S$ is a planar set of spanning trees. Since $u_{\gamma}$ and $u'_{\gamma}$ are on the same side of ${\cal C}_{\beta}$, it follows that $u_{\alpha}$ is on one side of ${\cal C}_{\beta}$ (call it {\em the small side} of ${\cal C}_{\beta}$), while $v_{\alpha}$, $u'_{\alpha}$, and $v'_{\alpha}$ are on the other side (call it {\em the large side} of ${\cal C}_{\beta}$). Analogously, the cycle ${\cal C}_{\gamma}$ composed of $P_{\gamma}$ and $e_{\gamma}$ has $v_{\alpha}$ on one side (call it {\em the small side} of ${\cal C}_{\gamma}$), and $u_{\alpha}$, $u'_{\alpha}$, and $v'_{\alpha}$ on the other side (call it {\em the large side} of ${\cal C}_{\gamma}$). Observe that the small side of ${\cal C}_{\beta}$ and the small side of ${\cal C}_{\gamma}$ are disjoint, as otherwise $P_{\beta}$ intersects ${\cal C}_{\gamma}$ or $P_{\gamma}$ intersects ${\cal C}_{\beta}$.

Now consider the cycle ${\cal D}$ composed of $P_{\beta}$, of $P_{\gamma}$, of the path $P_u$ in $\cal C$ between $u_{\beta}$ and $u_{\gamma}$ not containing $u_{\alpha}$, and of the path $P_v$ in $\cal C$ between $v_{\beta}$ and $v_{\gamma}$ not containing $v_{\alpha}$. Cycle ${\cal D}$ contains vertices of $\alpha$ on both sides. Namely, it contains $u_{\alpha}$ and $v_{\alpha}$ on one side (the side of ${\cal D}$ containing the small side of ${\cal C}_{\beta}$ and the small side of ${\cal C}_{\gamma}$), and $u'_{\alpha}$ and $v'_{\alpha}$ on the other side. However, no con-edge $g_{\alpha}$ for $\alpha$ in $S$ crosses ${\cal D}$. In fact, $g_{\alpha}$ cannot cross $P_{\beta}$ and $P_{\gamma}$, as such paths are composed of con-edges in $S$, and it cannot cross $P_u$ and $P_v$ by construction of ${\cal C}$. It follows that $S$ does not connect $\alpha$, a contradiction to the fact that $S$ is a planar set of spanning trees.
\end{proof}



\begin{lemma}[{\sc Test 4}]\label{le:ababa}
Suppose that con-edges $e'_{\alpha},e''_{\alpha} \in A$ for $\alpha$ exist in ${\cal C}_{\alpha}$, and such that $e_{\alpha}$, $e''_{\alpha}$, and $e'_{\alpha}$ occur in this order along ${\cal C}_{\alpha}$, when clockwise traversing ${\cal C}_{\alpha}$. Suppose also that there exist con-edges $e'_{\beta},e''_{\beta} \in A$ for $\beta$ and $e'_{\gamma}\in A$ for $\gamma$ such that $e'_{\alpha}\conf e'_{\beta}$, $e'_{\alpha}\conf e'_{\gamma}$, and $e''_{\alpha}\conf e''_{\beta}$. Then, no planar set of spanning trees for $A$ exists.
\end{lemma}

\begin{proof}
Denote by $u_{\alpha}$ and $v_{\alpha}$ (by $u_{\beta}$ and $v_{\beta}$, by $u_{\gamma}$ and $v_{\gamma}$) the end-vertices of $e_{\alpha}$ (resp.\ of $e_{\beta}$, resp.\ of $e_{\gamma}$). By Property~\ref{pr:no-two-edges-same-structure}, it is possible to draw a closed curve $\cal C$ passing through $u_{\alpha}$, $u_{\beta}$, $u_{\gamma}$, $v_{\alpha}$, $v_{\gamma}$, and $v_{\beta}$ in this clockwise order, containing edges $e_{\alpha}$, $e_{\beta}$, and $e_{\gamma}$ in its interior, and containing every other con-edge for $\alpha$, $\beta$, and $\gamma$ in its exterior.

Denote by $u'_{\alpha}$ and $v'_{\alpha}$ (by $u'_{\beta}$ and $v'_{\beta}$, by $u'_{\gamma}$ and $v'_{\gamma}$) the end-vertices of $e'_{\alpha}$ (resp.\ of $e'_{\beta}$, resp.\ of $e'_{\gamma}$), and by $u''_{\alpha}$ and $v''_{\alpha}$ (by $u''_{\beta}$ and $v''_{\beta}$) the end-vertices of $e''_{\alpha}$ (resp.\ of $e''_{\beta}$). By Property~\ref{pr:no-two-edges-same-structure}, it is possible to draw a closed curve ${\cal C}'$ passing through $u'_{\alpha}$, $u'_{\gamma}$, $u'_{\beta}$, $v'_{\alpha}$, $v'_{\beta}$, and $v'_{\gamma}$ in this clockwise order, containing edges $e'_{\alpha}$, $e'_{\beta}$, and $e'_{\gamma}$ in its interior, and containing every other con-edge for $\alpha$, $\beta$, and $\gamma$ in its exterior. Also, it is possible to draw a closed curve ${\cal C}''$ passing through $u''_{\alpha}$, $u''_{\beta}$, $v''_{\alpha}$, and $v''_{\beta}$ in this clockwise order, containing edges $e''_{\alpha}$ and $e''_{\beta}$ in its interior, and containing every other con-edge for $\alpha$ and $\beta$ in its exterior. We assume that ${\cal C}''$ is arbitrarily close to the drawing of $e''_{\alpha}$ and $e''_{\beta}$ so that ${\cal C}''$ intersects a con-edge for a cluster different from $\alpha$ and $\beta$ only if that con-edge intersects $e''_{\alpha}$ or $e''_{\beta}$. Assume, w.l.o.g., that the face of $A[\alpha]$ delimited by ${\cal C}_{\alpha}$ is to the right of ${\cal C}_{\alpha}$ when traversing such a cycle in clockwise direction. Finally assume, w.l.o.g., that $u_{\alpha}$, $v_{\alpha}$, $u''_{\alpha}$, $v''_{\alpha}$, $u'_{\alpha}$, and $v'_{\alpha}$ appear in this clockwise order along ${\cal C}_{\alpha}$ (possibly $v_{\alpha}=u''_{\alpha}$, $v''_{\alpha}=u'_{\alpha}$, and/or $v'_{\alpha}=u_{\alpha}$).

First, suppose that $e''_{\alpha}$ has a conflict with a con-edge $e''_{\gamma}$ for $\gamma$. Then, if $e''_{\alpha}$ is crossed first by $e''_{\beta}$ and then by $e''_{\gamma}$ when ${\cal C}_{\alpha}$ is traversed clockwise, we can conclude that no planar set of spanning trees for $A$ exists by Lemma~\ref{le:same-order} (with $e_{\alpha}$ and $e''_{\alpha}$ playing the role of the edges $e_{\alpha}$ and $e'_{\alpha}$ in the statement of Lemma~\ref{le:same-order}). Analogously, if $e''_{\alpha}$ is crossed first by $e''_{\gamma}$ and then by $e''_{\beta}$ when ${\cal C}_{\alpha}$ is traversed clockwise, we can conclude that no planar set of spanning trees for $A$ exists by Lemma~\ref{le:same-order} (with $e'_{\alpha}$ and $e''_{\alpha}$ playing the role of the edges $e_{\alpha}$ and $e'_{\alpha}$ in the statement of Lemma~\ref{le:same-order}). Thus, in what follows we assume that $e''_{\alpha}$ does not have a conflict with any con-edge for $\gamma$.

Suppose, for a contradiction, that there exists a planar set $S$ of spanning trees for $A$. The proof distinguishes three cases.

\paragraph{Case 1.} Let $e_{\alpha},e'_{\alpha},e''_{\alpha} \notin S$. Refer to Fig.~\ref{fig:ababa}(a).

Consider the path $P^1_{\alpha}$ ($P^2_{\alpha}$, $P^3_{\alpha}$) all of whose edges belong to $S$ connecting $u_{\alpha}$ and $v''_{\alpha}$ (resp.\ $v_{\alpha}$ and $u'_{\alpha}$, resp.\ $u''_{\alpha}$ and $v'_{\alpha}$). Since the end-vertices of $P^i_{\alpha}$ and $P^j_{\alpha}$ alternate along ${\cal C}_{\alpha}$, for every $i\neq j$ with $i,j\in\{1,2,3\}$, and since ${\cal C}_{\alpha}$ delimits a face of $A[\alpha]$, it follows that $P^i_{\alpha}$ and $P^j_{\alpha}$ share vertices (and possibly edges). Thus, the union of $P^1_{\alpha}$, $P^2_{\alpha}$, and $P^3_{\alpha}$ is a tree $T_{\alpha}$ whose leaves can only be from $\{u_{\alpha},v_{\alpha},u'_{\alpha},v'_{\alpha},u''_{\alpha},v''_{\alpha}\}$. Consider the path $P^u_{\beta}$ all of whose edges belong to $S$ connecting $u_{\beta}$ and $u''_{\beta}$. We claim that this path contains $v_{\beta}$. Indeed, if $P^u_{\beta}$ does not contain $v_{\beta}$, then it crosses the path connecting $u_{\alpha}$ and $v_{\alpha}$ in $T_{\alpha}$, thus contradicting the fact that $S$ is a planar set of spanning trees for $A$. Analogously, $P^u_{\beta}$ contains $v''_{\beta}$. Further, $P^u_{\beta}$ does not contain $u'_{\beta}$, as otherwise it would cross the path connecting $u'_{\alpha}$ and $v'_{\alpha}$ in $T_{\alpha}$. Next, consider the path $P''_a$ that is the path in ${\cal C}''$ between $v''_{\alpha}$ and $u''_{\beta}$ and not containing $u''_{\alpha}$. Also, consider the path $P''_b$ that is the path in ${\cal C}''$ between $v''_{\alpha}$ and $v''_{\beta}$ and not containing $u''_{\alpha}$. Not both $P''_a$ and $P''_b$ intersect a con-edge for $\gamma$, as otherwise by construction of ${\cal C}''$ con-edge $e''_{\alpha}$ would have a conflict with a con-edge for $\gamma$, which contradicts the assumptions. Assume that $P''_a$ does not cross any con-edge for $\gamma$, the other case being analogous.

Consider the cycle ${\cal D}$ composed of $P^1_{\alpha}$, of $P^u_{\beta}$, of $P''_a$, and of the path $P_u$ in $\cal C$ between $u_{\alpha}$ and $u_{\beta}$ not containing $u_{\gamma}$. Cycle ${\cal D}$ contains vertices of $\gamma$ on both sides (e.g., $u_{\gamma}$ and $u'_{\gamma}$). However, no con-edge $g_{\gamma}$ for $\gamma$ in $S$ crosses ${\cal D}$. In fact, $g_{\gamma}$ cannot cross $P^1_{\alpha}$ or $P^u_{\beta}$, as such paths are composed of con-edges in $S$, it cannot cross $P_u$ by construction of $\cal C$, and it cannot cross $P''_a$ by assumption. It follows that $S$ does not connect $\gamma$, a contradiction to the fact that $S$ is a planar set of spanning trees for $A$.

\begin{figure}[tb]
\begin{center}
\begin{tabular}{c c}
\mbox{\includegraphics[scale=0.38]{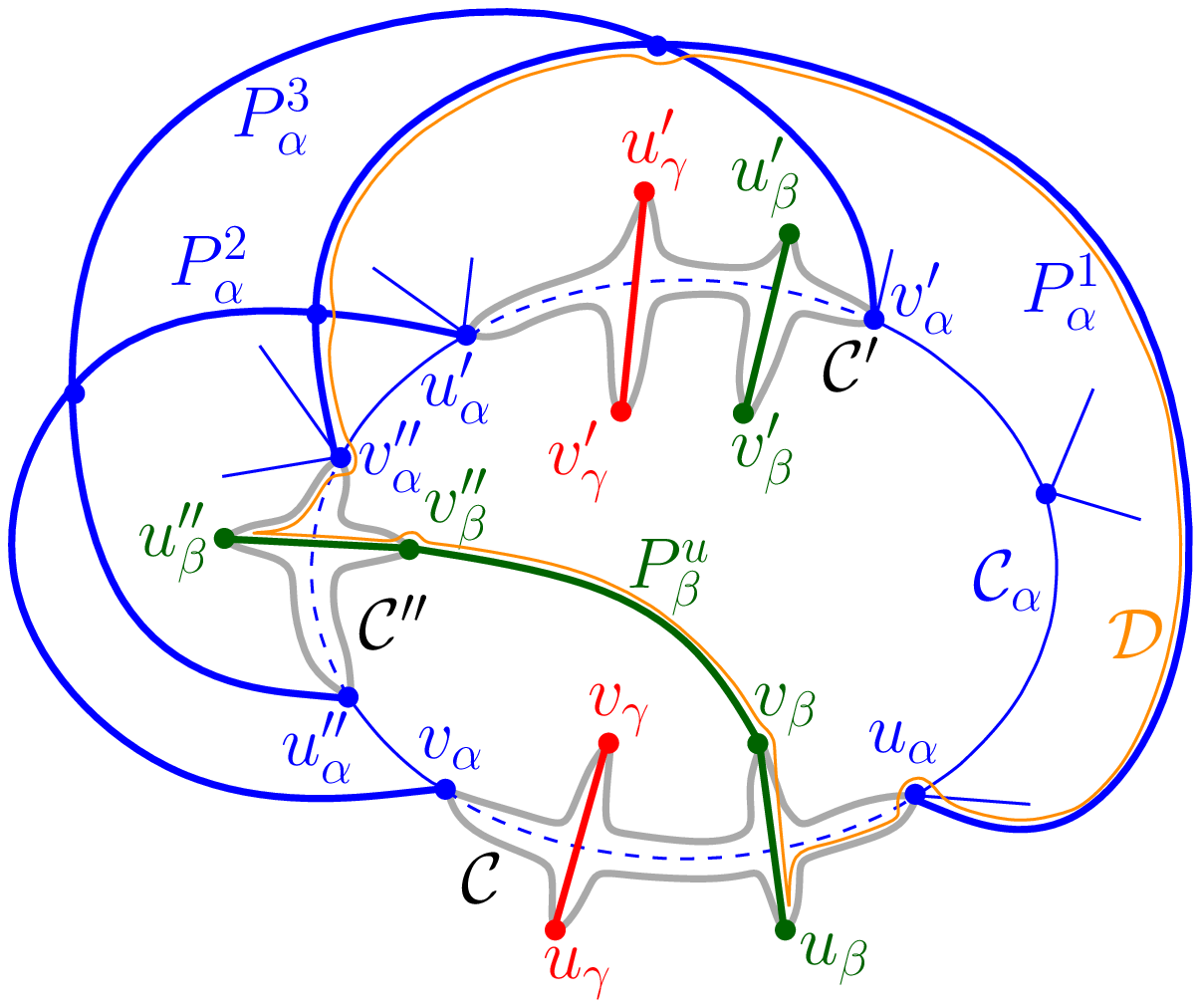}} \hspace{1mm} &
\mbox{\includegraphics[scale=0.38]{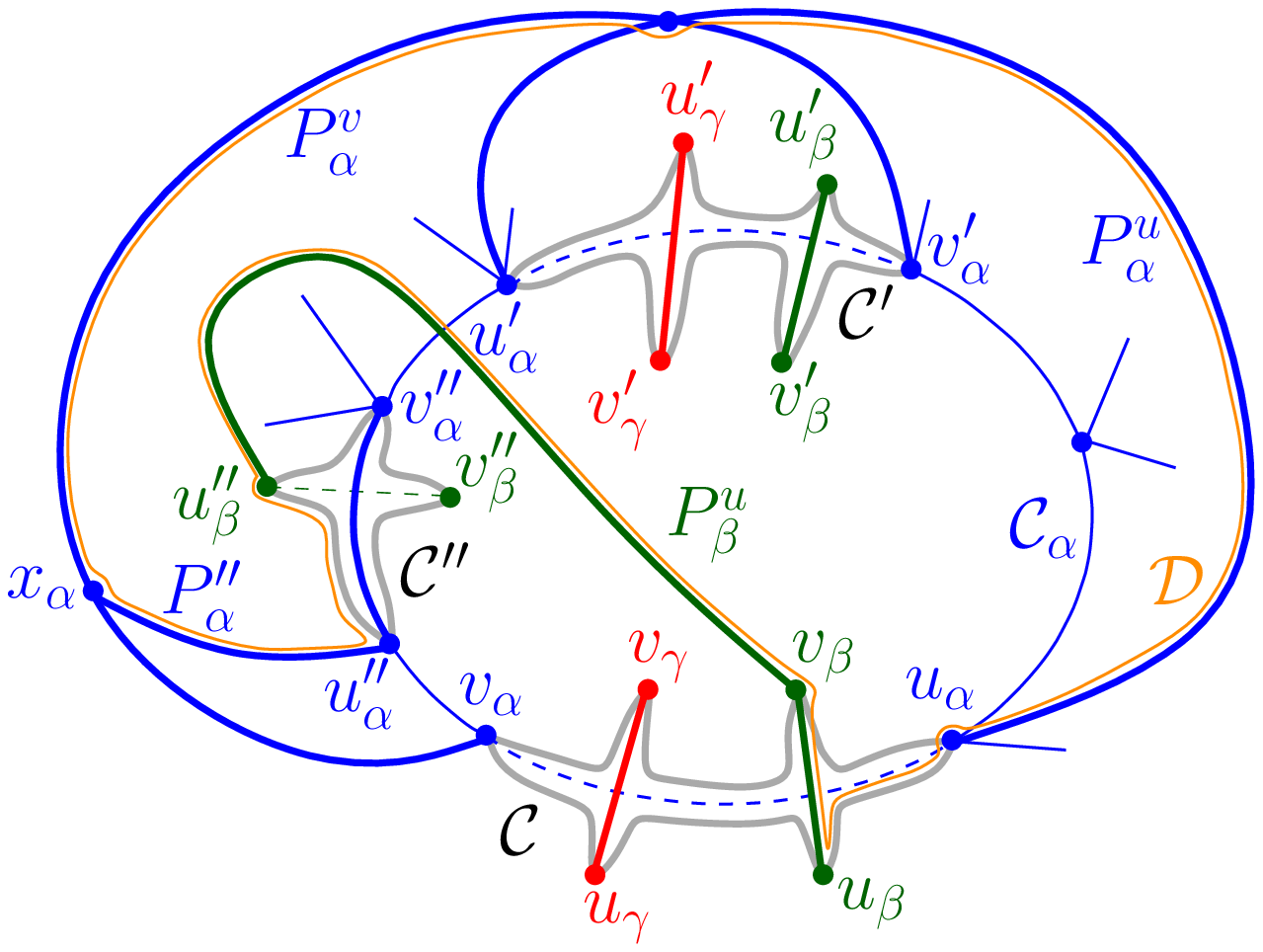}}\\
(a) \hspace{1mm} & (b)
\end{tabular}
\caption{Proof of Lemma~\ref{le:ababa}. (a) The case in which neither $e_{\alpha}$, nor $e'_{\alpha}$, nor $e''_{\alpha}$ belongs to $S$. (b) The case in which $e_{\alpha}$ and $e'_{\alpha}$ do not belong to $S$, while $e''_{\alpha}$ belongs to $S$.}
\label{fig:ababa}
\end{center}
\end{figure}

\paragraph{Case 2.} Let $e_{\alpha},e'_{\alpha}\notin S$, and $e''_{\alpha}\in S$. Refer to Fig.~\ref{fig:ababa}(b).

Consider the path $P^u_{\alpha}$ ($P^v_{\alpha}$) all of whose edges belong to $S$ connecting $u_{\alpha}$ and $u'_{\alpha}$ (resp.\ $v_{\alpha}$ and $v'_{\alpha}$). Since the end-vertices of $P^u_{\alpha}$ and $P^v_{\alpha}$ alternate along ${\cal C}_{\alpha}$ and since ${\cal C}_{\alpha}$ delimits a face of $A[\alpha]$, it follows that $P^u_{\alpha}$ and $P^v_{\alpha}$ share vertices (and possibly edges). Thus, the union of $P^u_{\alpha}$ and $P^v_{\alpha}$ is a tree $T_{\alpha}$ whose leaves can only be from $\{u_{\alpha},v_{\alpha},u'_{\alpha},v'_{\alpha}\}$. Also, let $P''_{\alpha}$ be a path whose edges belong to $S$ connecting $u''_{\alpha}$ and any vertex in $T_{\alpha}$, say $x_\alpha$. Assume that $P''_{\alpha}$ is minimal, i.e., no vertex of $P''_{\alpha}$ different from $x_\alpha$ belongs to $T_{\alpha}$.
Next, consider the path $P''_a$ in ${\cal C}''$ between $u''_{\alpha}$ and $u''_{\beta}$ and not containing $v''_{\alpha}$. Also, consider the path $P''_b$ in ${\cal C}''$ between $u''_{\alpha}$ and $v''_{\beta}$ and not containing $v''_{\alpha}$. Not both $P''_a$ and $P''_b$ intersect a con-edge for $\gamma$, as otherwise by construction of ${\cal C}''$ con-edge $e''_{\alpha}$ would have a conflict with a con-edge for $\gamma$, which contradicts the assumptions. Assume that $P''_a$ does not cross any con-edge for $\gamma$, the other case being analogous.
Now consider the path $P^u_{\beta}$ all of whose edges belong to $S$ connecting $u_{\beta}$ and $u''_{\beta}$. We claim that this path contains $v_{\beta}$. Indeed, if $P^u_{\beta}$ does not contain $v_{\beta}$, then it crosses the path connecting $u_{\alpha}$ and $v_{\alpha}$ in $T_{\alpha}$, thus contradicting the fact that $S$ is a planar set of spanning trees for $A$.

Consider the cycle ${\cal D}$ composed of $P^u_{\beta}$, of the path $P_u$ in $\cal C$ between $u_{\alpha}$ and $u_{\beta}$ not containing $u_{\gamma}$, of the path connecting $u_{\alpha}$ and $x_{\alpha}$ in $T_{\alpha}$, of $P''_{\alpha}$, and of $P''_a$. Cycle ${\cal D}$ contains vertices of $\gamma$ on both sides (e.g., $u_{\gamma}$ and $u'_{\gamma}$). However, no con-edge $g_{\gamma}$ for $\gamma$ in $S$ can cross ${\cal D}$. In fact, $g_{\gamma}$ cannot cross $P^u_{\beta}$, the path connecting $u_{\alpha}$ and $x_{\alpha}$ in $T_{\alpha}$, or $P''_{\alpha}$, as such paths are composed of con-edges in $S$, it cannot cross $P_u$ by construction of $\cal C$, and it cannot cross $P''_a$ by assumption. It follows that $S$ does not connect $\gamma$, a contradiction to the fact that $S$ is a planar set of spanning trees for $A$.

\begin{figure}[tb]
\begin{center}
\mbox{\includegraphics[scale=0.38]{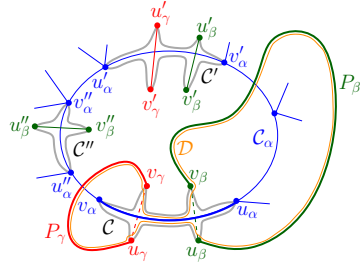}}
\caption{Proof of Lemma~\ref{le:ababa}, in the case in which $e_{\alpha}$ belongs to $A$.}
\label{fig:ababa-2}
\end{center}
\end{figure}

\paragraph{Case 3.} Let $\{e_{\alpha},e'_{\alpha}\}\cap S\neq \emptyset$. Suppose, w.l.o.g., that $e_{\alpha}\in S$. Refer to~Fig.~\ref{fig:ababa-2}. This case is similar to the second case in the proof of Lemma~\ref{le:same-order}.

Consider the path $P_{\beta}$ ($P_{\gamma}$) all of whose edges belong to $S$ connecting $u_{\beta}$ and $v_{\beta}$ (resp.\ $u_{\gamma}$ and $v_{\gamma}$).
Consider the cycle ${\cal C}_{\beta}$ composed of $P_{\beta}$ and $e_{\beta}$. We have that no con-edge $g_{\gamma}$ for $\gamma$ in $S$ crosses ${\cal C}_{\beta}$. In fact, $g_{\gamma}$ cannot cross $P_{\beta}$, as such a path is composed of con-edges in $S$, and it cannot cross $e_{\beta}$ by Property~\ref{pr:no-two-edges-same-structure}, given that $e_{\beta}$ and $e_{\gamma}$ belong to the same connected component of $K_A$ and do not cross, as otherwise {\sc Test 2} would succeed on $A$. It follows that ${\cal C}_{\beta}$ has $u_{\gamma}$ and $u'_{\gamma}$ on the same side, as otherwise $S$ would not connect $\gamma$, a contradiction to the fact that $S$ is a planar set of spanning trees for $A$.  Since $u_{\gamma}$ and $u'_{\gamma}$ are on the same side of ${\cal C}_{\beta}$, it follows that $u_{\alpha}$ is on one side of ${\cal C}_{\beta}$ (call it {\em the small side} of ${\cal C}_{\beta}$), while $v_{\alpha}$ and $u'_{\alpha}$ are on the other side (call it {\em the large side} of ${\cal C}_{\beta}$). Observe that, differently from the proof of Lemma~\ref{le:same-order}, it might be the case that $v'_{\alpha}$ is in the small side of ${\cal C}_{\beta}$, if $P_{\beta}$ contains con-edge $(u'_{\beta},v'_{\beta})$.
An analogous argument proves that the cycle ${\cal C}_{\gamma}$ composed of $P_{\gamma}$ and $e_{\gamma}$ has $v_{\alpha}$ on one side (call it {\em the small side} of ${\cal C}_{\gamma}$), and $u_{\alpha}$, $u'_{\alpha}$, and $v'_{\alpha}$ on the other side (call it {\em the large side} of ${\cal C}_{\gamma}$).
Observe that the small side of ${\cal C}_{\beta}$ and the small side of ${\cal C}_{\gamma}$ are disjoint, as otherwise $P_{\beta}$ intersects ${\cal C}_{\gamma}$ or $P_{\gamma}$ intersects ${\cal C}_{\beta}$.

Now consider the cycle ${\cal D}$ composed of $P_{\beta}$, of $P_{\gamma}$, of the path $P_u$ in $\cal C$ between $u_{\beta}$ and $u_{\gamma}$ not containing $u_{\alpha}$, and of the path $P_v$ in $\cal C$ between $v_{\beta}$ and $v_{\gamma}$ not containing $v_{\alpha}$. Cycle ${\cal D}$ contains vertices of $\alpha$ on both sides. Namely, it contains $u_{\alpha}$ and $v_{\alpha}$ on one side (the side of ${\cal D}$ containing the small side of ${\cal C}_{\beta}$ and the small side of ${\cal C}_{\gamma}$), and $u'_{\alpha}$ on the other side. However, no con-edge $g_{\alpha}$ for $\alpha$ in $S$ can cross ${\cal D}$. In fact, $g_{\alpha}$ cannot cross $P_{\beta}$ and $P_{\gamma}$, as such paths are composed of con-edges in $S$, and it cannot cross $P_u$ and $P_v$ by construction of $\cal C$. It follows that $S$ does not connect $\alpha$, a contradiction to the fact that $S$ is a planar set of spanning trees for $A$.
\end{proof}


If {\sc Simplifications 1--4} do not apply to $A$ and {\sc Tests 1--4} fail on $A$, then the con-edges for a cluster $\alpha$ that are crossed by con-edges for (at least) two other clusters have a nice structure, that we call {\em $\alpha$-donut} (see Fig.~\ref{fig:donut}). 

Consider a con-edge $e_{\alpha}\in A$ for $\alpha$ crossing con-edges $e_{\beta_1},\dots,e_{\beta_m}$ for clusters $\beta_1,\dots,\beta_m$, with $m\geq 2$. An $\alpha$-donut for $e_{\alpha}$ consists of a sequence $e^1_\alpha,\dots,e^{k}_\alpha,e^{k+1}_\alpha$ of con-edges for $\alpha$ with $k\geq 2$, called {\em spokes} of the $\alpha$-donut, of a sequence ${\cal C}^1_\alpha,\dots,{\cal C}^{k}_\alpha,{\cal C}^{k+1}_\alpha={\cal C}^{1}_\alpha$ of facial cycles in $A[\alpha]$, and of sequences $e^1_{\beta_j},\dots,e^k_{\beta_j}$ of con-edges for $\beta_j$, for each $1\leq j\leq m$, such that the following hold for every $1\leq i\leq k$: 

\begin{itemize}
\item[(a)] $e_{\alpha}$ is one of edges $e^1_\alpha,\dots,e^{k}_\alpha$;
\item[(b)] $e^i_\alpha\conf e^i_{\beta_j}$, for every $1\leq j\leq m$;
\item[(c)] ${\cal C}^i_\alpha$ and ${\cal C}^{i+1}_\alpha$ share edge $e^i_\alpha$;
\item[(d)] edge $e^i_{\alpha}$ is crossed by $e^i_{\beta_1},\dots,e^i_{\beta_m}$ in this order when ${\cal C}^i_{\alpha}$ is traversed clockwise;
\item[(e)] all the con-edges of ${\cal C}^{i+1}_{\alpha}$ encountered when clockwise traversing ${\cal C}^{i+1}_{\alpha}$ from $e^{i}_{\alpha}$ to $e^{i+1}_{\alpha}$ do not cross any con-edge for $\beta_2,\dots,\beta_m$; and
\item[(f)] all the con-edges of ${\cal C}^{i+1}_{\alpha}$ encountered when clockwise traversing ${\cal C}^{i+1}_{\alpha}$ from $e^{i+1}_{\alpha}$ to $e^{i}_{\alpha}$ do not cross any con-edge for $\beta_1,\dots,\beta_{m-1}$. 
\end{itemize}

We have the following.

\begin{lemma} \label{le:donut}
For every con-edge $e_{\alpha}\in A$ for $\alpha$, there exists an $\alpha$-donut for $e_{\alpha}$.
\end{lemma}

\begin{figure}[tb]
\begin{center}
\mbox{\includegraphics[scale=0.4]{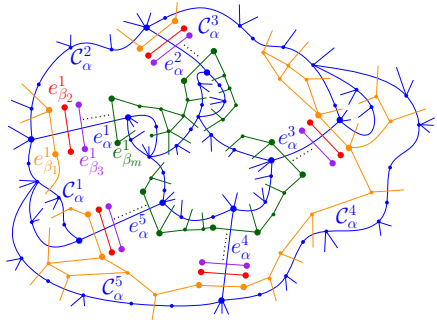}}
\caption{The $\alpha$-donut for $e_{\alpha}$. Only the con-edges of ${\cal C}^1_\alpha,\dots,{\cal C}^{k}_\alpha$, the con-edges for $\beta_1,\dots,\beta_m$ crossing the spokes of the $\alpha$-donut for $\alpha$, the con-edges for $\beta_1$ and $\beta_m$ inside the faces delimited by ${\cal C}^1_\alpha,\dots,{\cal C}^{k}_\alpha$, part of the con-edges for $\alpha$ incident to vertices in ${\cal C}^1_\alpha,\dots,{\cal C}^{k}_\alpha$, and part of the con-edges of $\beta_1$ and $\beta_m$ crossing the con-edges of ${\cal C}^1_\alpha,\dots,{\cal C}^{k}_\alpha$ are shown.}
\label{fig:donut}
\end{center}
\end{figure}

\begin{proof}
Let, w.l.o.g., $e^1_{\alpha}=e_{\alpha}$ and consider the two faces $f^1_{\alpha}$ and $f^2_{\alpha}$ of $A[\alpha]$ incident to $e^1_{\alpha}$. Since {\sc Simplification 1} does not apply to $A$, it follows that $f^1_{\alpha}\neq f^2_{\alpha}$. Let ${\cal C}^i_{\alpha}$ be the cycle delimiting $f^i_{\alpha}$, for $i=1,2$. Let $e^1_{\beta_1},\dots,e^1_{\beta_m}$ be the con-edges for clusters $\beta_1,\dots,\beta_m$, respectively, ordered as they cross $e^1_{\alpha}$ when clockwise traversing ${\cal C}^1_{\alpha}$. Thus, $e^1_{\alpha}$ is crossed by $e^1_{\beta_m},\dots,e^1_{\beta_1}$ in this order when ${\cal C}^2_{\alpha}$ is traversed clockwise.

Consider facial cycle ${\cal C}^2_{\alpha}$.

Since {\sc Simplification 4} does not apply to $A$, there exists at least one con-edge $e^2_{\alpha}$ in ${\cal C}^2_{\alpha}$ that is different from $e^1_{\alpha}$ and that is crossed by con-edges $e^2_{\beta_1}$ for $\beta_1$ and $e^2_{\beta_2}$ for $\beta_2$. Since {\sc Test 3} fails on $A$, it follows that $e^2_{\alpha}$ is crossed first by $e^2_{\beta_1}$ and then by $e^2_{\beta_2}$ when clockwise traversing ${\cal C}^2_{\alpha}$. Since {\sc Test 4} fails on $A$, it follows that all the con-edges of ${\cal C}^2_{\alpha}$ different from $e^1_{\alpha}$ and $e^{2}_{\alpha}$ encountered when clockwise traversing ${\cal C}^2_{\alpha}$ from $e^1_{\alpha}$ to $e^{2}_{\alpha}$ (from $e^2_{\alpha}$ to $e^1_{\alpha}$) do not have a conflict with any con-edge for $\beta_2$ (resp.\ for $\beta_1$).

\begin{figure}[tb]
\begin{center}
\begin{tabular}{c c}
\mbox{\includegraphics[scale=0.38]{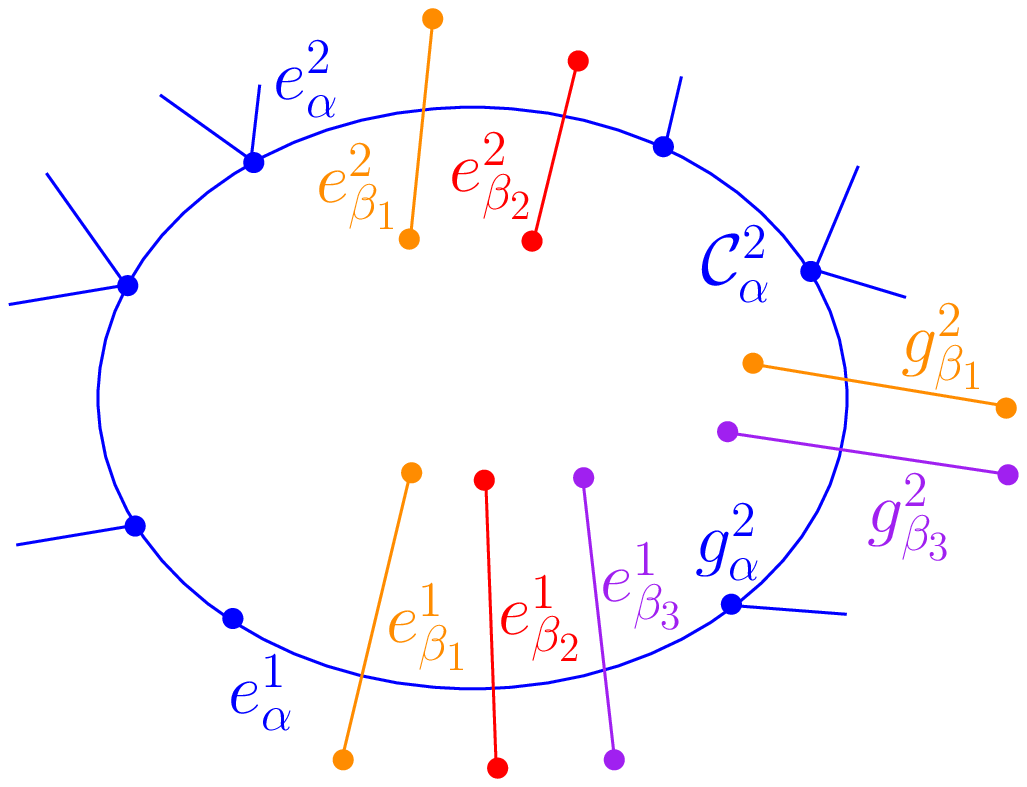}} \hspace{1cm} &
\mbox{\includegraphics[scale=0.38]{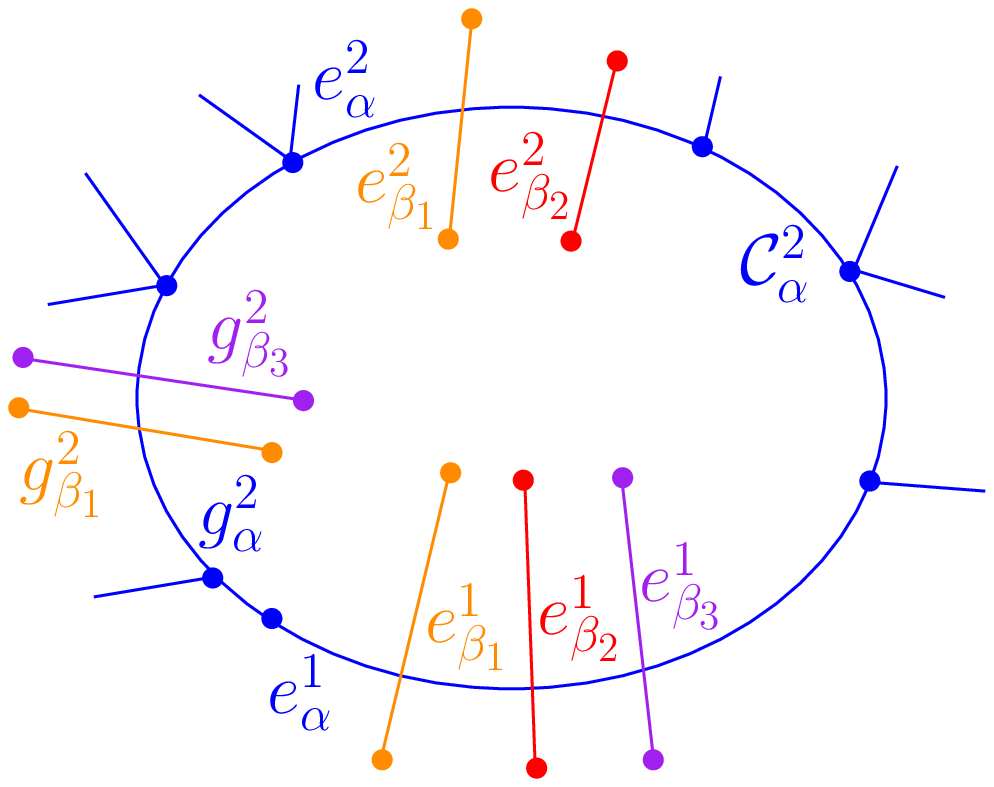}}\\
(a) \hspace{1cm} & (b)
\end{tabular}
\caption{Illustration for the proof of Lemma~\ref{le:donut}.}
\label{fig:donut-proof}
\end{center}
\end{figure}

Now, since {\sc Simplification 4} does not apply to $A$, there exists at least one con-edge $g^2_{\alpha}$ in ${\cal C}^2_{\alpha}$ that is different from $e^1_{\alpha}$ and that is crossed by con-edges $g^2_{\beta_1}$ for $\beta_1$ and $g^2_{\beta_3}$ for $\beta_3$. We prove that $g^2_{\alpha}=e^2_{\alpha}$. Suppose, for a contradiction, that $g^2_{\alpha}\neq e^2_{\alpha}$. If $g^2_{\alpha}$ is encountered when clockwise traversing ${\cal C}^2_{\alpha}$ from $e^2_{\alpha}$ to $e^1_{\alpha}$, as in Fig.~\ref{fig:donut-proof}(a), then {\sc Test 4} would succeed on $A$, with $\alpha$, $\beta_1$, and $\beta_2$ playing the roles of $\alpha$, $\beta$, and $\gamma$, respectively, in the statement of Lemma~\ref{le:ababa}, a contradiction. Hence, assume that $g^2_{\alpha}$ is encountered when clockwise traversing ${\cal C}^2_{\alpha}$ from $e^1_{\alpha}$ to $e^2_{\alpha}$, as in Fig.~\ref{fig:donut-proof}(b). Since {\sc Test 3} fails on $A$, it follows that $g^2_{\alpha}$ is crossed first by $g^2_{\beta_1}$ and then by $g^2_{\beta_3}$ when clockwise traversing ${\cal C}^2_{\alpha}$. However, this implies that {\sc Test 4} succeeds on $A$, with $\alpha$, $\beta_1$, and $\beta_3$ playing the roles of $\alpha$, $\beta$, and $\gamma$, respectively, in the statement of Lemma~\ref{le:ababa}, a contradiction. Thus, we get that $g^2_{\alpha}=e^2_{\alpha}$, hence $e^2_{\alpha}$ is crossed by a con-edge $e^2_{\beta_3}$ for $\beta_3$. Since {\sc Test 3} fails on $A$, it follows that $e^2_{\alpha}$ is crossed first by $e^2_{\beta_1}$, then by $e^2_{\beta_2}$, and then by $e^2_{\beta_3}$ when clockwise traversing ${\cal C}^2_{\alpha}$. Since {\sc Test 4} fails on $A$, it follows that all the con-edges of ${\cal C}^2_{\alpha}$ different from $e^1_{\alpha}$ and $e^{2}_{\alpha}$ encountered when clockwise traversing ${\cal C}^2_{\alpha}$ from $e^1_{\alpha}$ to $e^{2}_{\alpha}$ (from $e^2_{\alpha}$ to $e^1_{\alpha}$) do not have a conflict with any con-edges for $\beta_2,\beta_3$ (resp.\ for $\beta_1,\beta_2$).

The argument in the previous paragraph can be repeated for each $\beta_j$, with $j=4,5,\dots,m$, with $\alpha$, $\beta_1$, $\beta_{j-1}$, and $\beta_j$ playing the roles of  $\alpha$, $\beta_1$, $\beta_2$, and $\beta_3$. This leads to conclude that $e^2_{\alpha}$ is crossed by con-edges $e^2_{\beta_1},\dots,e^2_{\beta_m}$ for $\beta_1,\dots,\beta_m$, respectively, in this order when clockwise traversing ${\cal C}^2_{\alpha}$, and that all the con-edges of ${\cal C}^2_{\alpha}$ encountered when clockwise traversing ${\cal C}^2_{\alpha}$ from $e^1_{\alpha}$ to $e^2_{\alpha}$ (from $e^2_{\alpha}$ to $e^1_{\alpha}$) do not have a conflict with any con-edges for $\beta_2,\dots,\beta_m$ (resp.\ for $\beta_1,\dots,\beta_{m-1}$).

Now the same argument as the one we just presented for ${\cal C}^2_{\alpha}$ is repeated for ${\cal C}^3_{\alpha}$, that is the facial cycle that contains $e^2_{\alpha}$ and that is different from ${\cal C}^2_{\alpha}$. Again, this leads to conclude that there exists a con-edge $e^3_{\alpha}\neq e^2_{\alpha}$ for $\alpha$ that belongs to ${\cal C}^3_{\alpha}$, that there exist con-edges $e^3_{\beta_1},\dots,e^3_{\beta_m}$ for clusters $\beta_1,\dots,\beta_m$, respectively, that cross $e^3_{\alpha}$ in this order when clockwise traversing ${\cal C}^3_{\alpha}$, and that all the con-edges of ${\cal C}^3_{\alpha}$ encountered when clockwise traversing ${\cal C}^3_{\alpha}$ from $e^2_{\alpha}$ to $e^3_{\alpha}$ (from $e^3_{\alpha}$ to $e^2_{\alpha}$) do not have a conflict with any con-edges for $\beta_2,\dots,\beta_m$ (resp.\ for $\beta_1,\dots,\beta_{m-1}$).

Since the number of edges of $A[\alpha]$ is finite and since each facial cycle of $A[\alpha]$ does not contain more than two con-edges crossed by con-edges for all of $\beta_1,\dots,\beta_m$ (as otherwise {\sc Test 3} would succeed on $A$), eventually a facial cycle ${\cal C}^{k+1}_{\alpha}={\cal C}^{1}_{\alpha}$ of $A[\alpha]$ is considered in which the two con-edges that are crossed by con-edges for all of $\beta_1,\dots,\beta_m$ are $e^k_{\alpha}$ and $e^1_{\alpha}$. This concludes the proof of the lemma.
\end{proof}

Observe that the $\alpha$-donut for any con-edge $e_{\alpha}$ for $\alpha$ can be computed efficiently. The following is a consequence of Lemma~\ref{le:donut}.

\begin{lemma} \label{le:exactly-one}
Consider a con-edge $e_{\alpha}$ for $\alpha$ that has a conflict with $m\geq 2$ con-edges for other clusters. Let $e^1_\alpha,\dots,e^{k}_\alpha$ be the spokes of the $\alpha$-donut for $e_{\alpha}$. Then, if a planar set $S$ of spanning trees for $A$ exists, it contains exactly one of $e^1_\alpha,\dots,e^{k}_\alpha$.
\end{lemma}

\begin{proof}
First, by Lemma~\ref{le:donut}, $e^i_\alpha$ belongs to both facial cycles ${\cal C}^i_\alpha$ and ${\cal C}^{i+1}_\alpha$, for every $1\leq i\leq k$, where ${\cal C}^{k+1}_\alpha={\cal C}^{1}_\alpha$. It follows that removing all of $e^1_\alpha,\dots,e^{k}_\alpha$ from $A$ disconnects $A[\alpha]$. Hence, $S$ contains at least one of $e^1_\alpha,\dots,e^{k}_\alpha$.

Suppose, for a contradiction, that a planar set $S$ of spanning trees for $A$ exists that contains at least two edges $e^x_\alpha$ and $e^y_\alpha$. Refer to Fig.~\ref{fig:exactly-one}. Denote by $u^x_\alpha$ and $v^x_\alpha$ (by $u^y_\alpha$ and $v^y_\alpha$) the end-vertices of $e^x_\alpha$ (resp.\ of $e^y_\alpha$), where we assume w.l.o.g. that edges $e^x_{\beta_1},e^x_{\beta_2},\dots,e^x_{\beta_m}$ are crossed in this order when traversing $e^x_\alpha$ from $u^x_\alpha$ to $v^x_\alpha$, and that edges $e^y_{\beta_1},e^y_{\beta_2},\dots,e^y_{\beta_m}$ are crossed in this order when traversing $e^y_\alpha$ from $u^y_\alpha$ to $v^y_\alpha$. Further, denote by $u^x_{\beta_1}$ and $v^x_{\beta_1}$, by $u^x_{\beta_m}$ and $v^x_{\beta_m}$, by $u^y_{\beta_1}$ and $v^y_{\beta_1}$, by $u^y_{\beta_m}$ and $v^y_{\beta_m}$,  the end-vertices of $e^x_{\beta_1}$, of $e^x_{\beta_m}$, of $e^y_{\beta_1}$, and of $e^y_{\beta_m}$, respectively.

\begin{figure}[tb]
\begin{center}
\mbox{\includegraphics[scale=0.45]{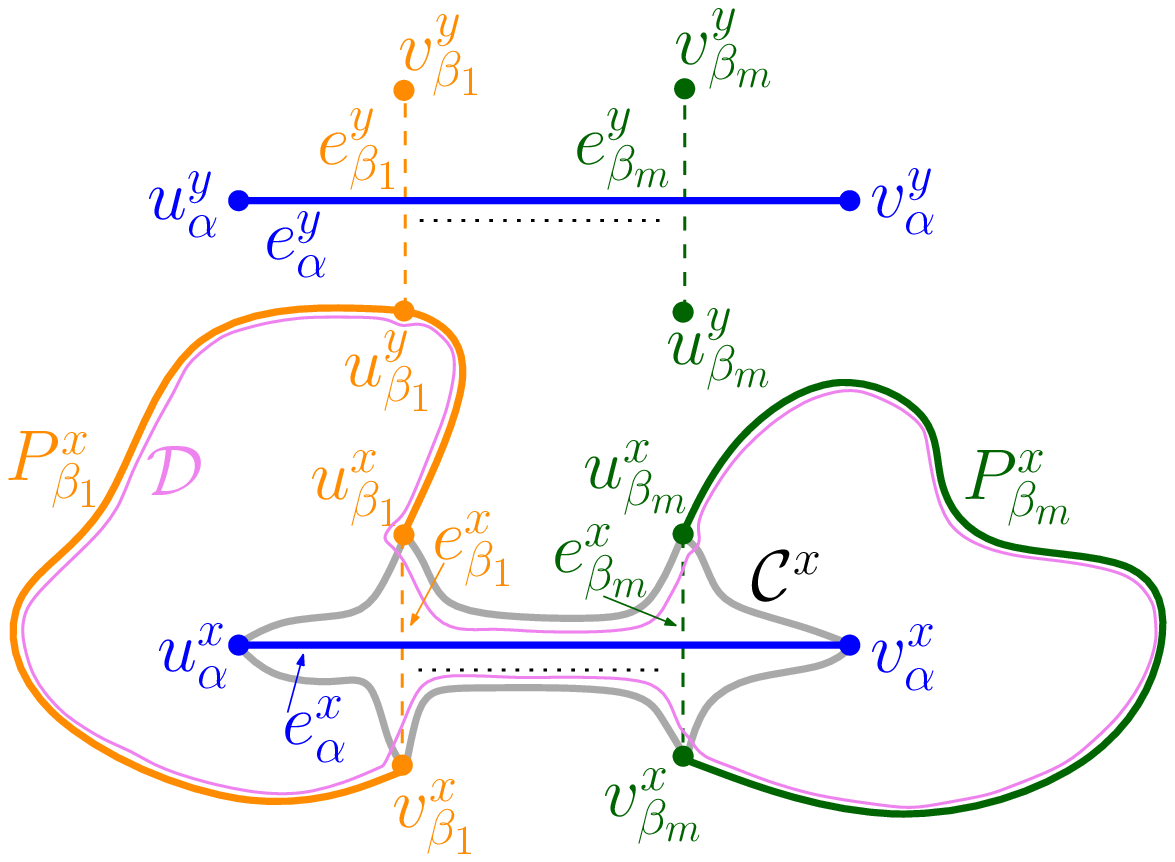}}
\caption{Illustration for the proof of Lemma~\ref{le:exactly-one}.}
\label{fig:exactly-one}
\end{center}
\end{figure}

Consider the path $P^x_{\beta_1}$ ($P^x_{\beta_m}$) all of whose edges belong to $S$ connecting $u^x_{\beta_1}$ and $v^x_{\beta_1}$ (resp.\ $u^x_{\beta_m}$ and $v^x_{\beta_m}$). Consider the cycle ${\cal C}^x_{\beta_1}$ composed of $P^x_{\beta_1}$ and $e^x_{\beta_1}$. We have that no con-edge $g_{\beta_m}$ for $\beta_m$ in $S$ crosses ${\cal C}^x_{\beta_1}$. In fact, $g_{\beta_m}$ cannot cross $P^x_{\beta_1}$, as such a path is composed of con-edges in $S$, and it cannot cross $e^x_{\beta_1}$ by Property~\ref{pr:no-two-edges-same-structure}, given that $e^x_{\beta_1}$ and $e^x_{\beta_m}$ belong to the same connected component of $K_A$ and do not cross, as otherwise {\sc Test 2} would succeed on $A$. It follows that ${\cal C}^x_{\beta_1}$ has $u^x_{\beta_m}$ and $u^y_{\beta_m}$ on the same side, as otherwise $S$ would not connect $\beta_m$, a contradiction to the fact that $S$ is a planar set of spanning trees for $A$.  Since $u^x_{\beta_m}$ and $u^y_{\beta_m}$ are on the same side of ${\cal C}^x_{\beta_1}$, since $u^y_{\beta_m}$, $u^y_{\alpha}$, and $v^y_{\alpha}$ are on the same side of ${\cal C}^x_{\beta_1}$, and since $u^x_{\beta_m}$ and $v^x_{\alpha}$ are on the same side of ${\cal C}^x_{\beta_1}$, it follows that $u^x_{\alpha}$ is on one side of ${\cal C}^x_{\beta_1}$ ({\em the small side} of ${\cal C}^x_{\beta_1}$), while $v^x_{\alpha}$, $u^y_{\alpha}$, and $v^y_{\alpha}$ are on the other side ({\em the large side} of ${\cal C}^x_{\beta_1}$). Analogously, the cycle ${\cal C}^x_{\beta_m}$ composed of $P^x_{\beta_m}$ and $e^x_{\beta_m}$ has $v^x_{\alpha}$ on one side ({\em the small side} of ${\cal C}^x_{\beta_m}$), and $u^x_{\alpha}$, $u^y_{\alpha}$, and $v^y_{\alpha}$ on the other side ({\em the large side} of ${\cal C}^x_{\beta_m}$). The small side of ${\cal C}^x_{\beta_1}$ and the small side of ${\cal C}^x_{\beta_m}$ are disjoint, as otherwise $P^x_{\beta_1}$ intersects ${\cal C}^x_{\beta_m}$, or $P^x_{\beta_m}$ intersects ${\cal C}^x_{\beta_1}$.

By Property~\ref{pr:no-two-edges-same-structure}, it is possible to draw a closed curve ${\cal C}^x$ passing through $u^x_\alpha$, $u^x_{\beta_1}$, $u^x_{\beta_m}$, $v^x_\alpha$, $v^x_{\beta_m}$, and $v^x_{\beta_1}$ in this circular order, containing edges $e^x_\alpha$, $e^x_{\beta_1}$, and $e^x_{\beta_m}$ in its interior, and containing every other con-edge for $\alpha$, $\beta_1$, and $\beta_m$ in its exterior. Now consider the cycle ${\cal D}$ composed of $P^x_{\beta_1}$, of $P^x_{\beta_m}$, of the path $P^x_u$ in ${\cal C}^x$ between $u^x_{\beta_1}$ and $u^x_{\beta_m}$ not containing $u^x_{\alpha}$, and of the path $P^x_v$ in ${\cal C}^x$ between $v^x_{\beta_1}$ and $v^x_{\beta_m}$ not containing $v^x_{\alpha}$. Cycle ${\cal D}$ contains vertices of $\alpha$ on both sides. Namely, it contains $u^x_{\alpha}$ and $v^x_{\alpha}$ on one side (the side of ${\cal D}$ containing the small side of ${\cal C}^x_{\beta_1}$ and the small side of ${\cal C}^x_{\beta_m}$), and $u^y_{\alpha}$ and $v^y_{\alpha}$ on the other side. However, no con-edge $g_{\alpha}$ for $\alpha$ in $S$ crosses ${\cal D}$. In fact, $g_{\alpha}$ cannot cross $P^x_{\beta_1}$ and $P^x_{\beta_m}$, as such paths are composed of con-edges in $S$, and it cannot cross $P^x_u$ and $P^x_v$ by construction of ${\cal C}^x$. It follows that $S$ does not connect $\alpha$, a contradiction to the fact that $S$ is a planar set of spanning trees for $A$.
\end{proof}

Consider a con-edge $e$ for a cluster $\alpha$. The {\em conflicting structure} $M(e)$ of $e$ is a sequence of sets $H_0(e), L_1(e), H_1(e), L_2(e), H_2(e), \dots$ of con-edges which correspond to the layers of a BFS traversal starting at $e$ of the connected component of $K_A$ containing $e$. That is: $H_0(e)=\{e\}$; then, for $i\geq 1$, $L_i(e)$ is the set of con-edges that cross con-edges in $H_{i-1}(e)$ and that are not in $L_{i-1}(e)$, and $H_i(e)$ is the set of con-edges that cross con-edges in $L_{i}(e)$ and that are not in $H_{i-1}(e)$.

We now study the conflicting structures of the spokes $e^1_\alpha,\dots,e^{k}_\alpha$ of the $\alpha$-donut for a con-edge $e_{\alpha}$ for $\alpha$. No two edges in a set $H_i(e_{\alpha})$ or in a set $L_i(e_{\alpha})$ have a conflict, as otherwise {\sc Test 2} would succeed. Also, by Lemma~\ref{le:one-or-the-other}, any planar set $S$ of spanning trees for $A$ contains either all the edges in $\bigcup_i H_i(e_{\alpha})$ or all the edges in $\bigcup_i L_i(e_{\alpha})$.

Assume that $e_{\alpha}$ has a conflict with at least two con-edges for other clusters. For any $1\leq i \leq k$, we say that $e^i_{\alpha}$ and $e^{i+1}_{\alpha}$ have {\em isomorphic conflicting structures} if $e^i_{\alpha}$ and $e^{i+1}_{\alpha}$ belong to isomorphic connected components of $K_A$ and if the vertices of these components that are in correspondence under the isomorphism represent con-edges for the same cluster. Formally, $e^i_{\alpha}$ and $e^{i+1}_{\alpha}$ have isomorphic conflicting structures if there exists a bijective mapping $\delta$ between the edges in $M(e^i_{\alpha})$ and the edges in $M(e^{i+1}_{\alpha})$ such that:
\begin{enumerate}
\item $e$ is a con-edge for a cluster $\varrho$ if and only if $\delta(e)$ is a con-edge for $\varrho$, for every $e\in M(e^i_{\alpha})$;
\item $e\in H_j(e^i_{\alpha})$ if and only if $\delta(e)\in H_j(e^{i+1}_{\alpha})$, for every $e\in M(e^i_{\alpha})$;
\item $e\in L_j(e^i_{\alpha})$ if and only if $\delta(e)\in L_j(e^{i+1}_{\alpha})$, for every $e\in M(e^i_{\alpha})$; and
\item $e\conf f$ if and only if $\delta(e)\conf \delta(f)$, for every $e,f\in M(e^i_{\alpha})$.
\end{enumerate}

Observe that the isomorphism of two conflicting structures can be tested efficiently.


We will prove in the following four lemmata that by examining the conflicting structures for the spokes of the $\alpha$-donut for $e_{\alpha}$, a decision on whether some spoke is or is not in $S$ can be taken without loss of generality. We start with the following:

\begin{lemma}[{\sc Simplification 5}]\label{le:isomorphic}
Suppose that spokes $e^i_{\alpha}$ and $e^{i+1}_{\alpha}$ have isomorphic conflicting structures. Then, there exists a planar set $S$ of spanning trees for $A$ if and only if there exists a planar set $S'$ of spanning trees for $A$ such that $e^i_{\alpha}\notin S'$.
\end{lemma}

\begin{proof}
If there exists no planar set of spanning trees for $A$, there is nothing to prove. Otherwise, consider any planar set $S$ of spanning trees for $A$. If $e^i_{\alpha}\notin S$, there is nothing to prove. Otherwise, suppose that $e^i_{\alpha}\in S$. Since $S$ does not contain any two con-edges that have a conflict and by Lemma~\ref{le:one-or-the-other}, we have $\bigcup_j H_j(e^i_{\alpha})\subseteq S$ and $S\cap \bigcup_j L_j(e^i_{\alpha})=\emptyset$.

By Lemma~\ref{le:exactly-one}, exactly one of $e^1_\alpha,\dots,e^{k}_\alpha$ belongs to any planar set $S$ of spanning trees for $A$. Hence, $e^{i+1}_{\alpha}\notin S$. Since $S$ does not contain any two con-edges that have a conflict and by Lemma~\ref{le:one-or-the-other}, we have $\bigcup_j L_j(e^{i+1}_{\alpha})\subseteq S$ and $S\cap \bigcup_j H_j(e^{i+1}_{\alpha})=\emptyset$.

Consider the set $S'\subseteq A$ of con-edges obtained from $S$ by removing $\bigcup_j H_j(e^i_{\alpha})$ and $\bigcup_j L_j(e^{i+1}_{\alpha})$ and by adding $\bigcup_j L_j(e^i_{\alpha})$ and $\bigcup_j H_j(e^{i+1}_{\alpha})$. We claim that $S'$ is a planar set of spanning trees for $A$. The claim directly implies the lemma.

First, we prove that no two con-edges in $S'$ have a conflict. Since $S$ is a planar set of spanning trees for $A$, no two con-edges in $S'\cap S$ have a conflict. Consider any con-edge $e\in L_j(e^i_{\alpha})$, for some $j\geq 1$, and consider any con-edge $g\in S'$. If $g$ does not belong to $M(e^i_{\alpha})$, then $e$ and $g$ do not cross, since $e$ and $e^i_{\alpha}$ belong to the same connected component of $K_A$. Further, if $g$ belongs to $M(e^i_{\alpha})$ and does not belong to $H_{j-1}(e^i_{\alpha})$ or to $H_{j}(e^i_{\alpha})$, then $e$ and $g$ do not cross, by definition of conflicting structure. Finally, $g$ does not belong to $H_{j-1}(e^i_{\alpha})$ or to $H_{j}(e^i_{\alpha})$, given that $g\in S'$. It can be analogously proved that no edge $e \in H_j(e^{i+1}_{\alpha})$, for some $j\geq 1$, crosses any con-edge $g\in S'$.

Second, we prove that, for each cluster $\mu$, the graph induced by the con-edges in $S'[\mu]$ is a tree that spans the vertices in $\mu$. This is trivially proved for every cluster $\mu$ that has no con-edge in $M(e^i_{\alpha})$, given that in this case $S'[\mu]=S[\mu]$. Moreover, by Property~\ref{pr:no-two-edges-same-structure} and since $M(e^i_{\alpha})$ and $M(e^{i+1}_{\alpha})$ are isomorphic, each cluster $\mu$ having a con-edge in $M(e^i_{\alpha})$ has {\em exactly} one con-edge $e^i_{\mu}$ in $M(e^i_{\alpha})$ and one con-edge $e^{i+1}_{\mu}$ in $M(e^{i+1}_{\alpha})$. By the construction of $S'$ and by the fact that if $e^i_{\mu}$ is in $H_j(e^i_{\alpha})$ (in $L_j(e^i_{\alpha})$), then $e^{i+1}_{\mu}$ is in $H_j(e^{i+1}_{\alpha})$ (resp.\ in $L_j(e^{i+1}_{\alpha})$), it follows that $S'[\mu]$ is obtained from $S[\mu]$ by removing $e^r_{\mu}$ and by adding $e^a_{\mu}$, for some distinct $a,r\in\{i,i+1\}$. Since $S[\mu]$ induces a spanning tree of the vertices in $\mu$, in order to prove that $S'[\mu]$ induces a spanning tree of the vertices in $\mu$ it suffices to prove that the end-vertices of $e^a_{\mu}$ belong to distinct connected components of $S[\mu]\setminus \{e^r_{\mu}\}$. In the following we prove this statement.

For $1\leq j\leq k$, denote by $u^j_\alpha$ and $v^j_\alpha$ the end-vertices of $e^j_\alpha$; assume w.l.o.g. that edges $e^j_{\beta_1},e^j_{\beta_2},\dots,e^j_{\beta_m}$ are crossed in this order when traversing $e^j_\alpha$ from $u^j_\alpha$ to $v^j_\alpha$. Denote by $u^j_{\beta_\ell}$ and $v^j_{\beta_\ell}$ the end-vertices of $e^j_{\beta_\ell}$, for every $1\leq j\leq k$ and $1\leq \ell\leq m$.

\begin{itemize}
\item We start with cluster $\alpha$. Denote by $S^1[\alpha]$ and $S^2[\alpha]$ the two connected components of $S[\alpha]$ obtained by removing $e^i_\alpha$ from $S[\alpha]$. Since $e^1_\alpha,\dots,e^{k}_\alpha$ are a separating set for $S[\alpha]$, since $e^i_\alpha$ is the only edge among $e^1_\alpha,\dots,e^{k}_\alpha$ that belongs to $S$ (by assumption and by Lemma~\ref{le:exactly-one}), it follows that $u^{i+1}_\alpha$ and $v^{i+1}_\alpha$ are one in $S^1[\alpha]$ and the other one in $S^2[\alpha]$.


\item We now deal with cluster $\beta_\ell$, for any $1\leq \ell\leq m$. Denote by $S^1[\beta_\ell]$ and $S^2[\beta_\ell]$ the two connected components of $S[\beta_\ell]$ obtained by removing $e^{i+1}_{\beta_\ell}$ from $S[\beta_\ell]$. In the following we prove that $u^i_{\beta_\ell}$ and $v^i_{\beta_\ell}$ are one in $S^1[\beta_\ell]$ and the other one in $S^2[\beta_\ell]$.

Suppose, for a contradiction, that both $u^i_{\beta_\ell}$ and $v^i_{\beta_\ell}$ are in $S^1[\beta_\ell]$. Then, there exists a path $P^i_{\beta_\ell}$ between $u^i_{\beta_\ell}$ and $v^i_{\beta_\ell}$ all of whose edges belong to $S^1[\beta_\ell]$. See Fig.~\ref{fig:distinct-components-2}.

\begin{figure}[tb]
\begin{center}
\mbox{\includegraphics[scale=0.38]{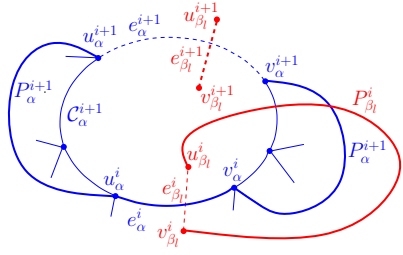}}
\caption{Illustration for the proof that the end-vertices of $e^{i}_{\beta_\ell}$ belong to distinct connected components of $S[\alpha]\setminus \{e^{i+1}_{\beta_\ell}\}$.}
\label{fig:distinct-components-2}
\end{center}
\end{figure}

Since $e^1_\alpha,\dots,e^{k}_\alpha$ are a separating set for $S[\alpha]$, since $e^i_\alpha\in S$, and since $e^{i+1}_\alpha\notin S$, it follows that the path $P^{i+1}_{\alpha}$ connecting $u^{i+1}_\alpha$ and $v^{i+1}_\alpha$ in $S$ contains edge $e^i_\alpha$. Since edges $e^i_\alpha$ and $e^{i}_{\beta_\ell}$ have a conflict, it follows that vertices $u^i_{\beta_\ell}$ and $v^i_{\beta_\ell}$ are on different sides of cycle $P^{i+1}_{\alpha}\cup e^{i+1}_{\alpha}$, hence path $P^i_{\beta_\ell}$ crosses cycle $P^{i+1}_{\alpha}\cup e^{i+1}_{\alpha}$ by the Jordan curve theorem. However, $P^i_{\beta_\ell}$ cannot cross $P^{i+1}_{\alpha}$, as all the edges of such paths belong to $S$, by assumption; moreover, $P^i_{\beta_\ell}$ cannot cross $e^{i+1}_{\alpha}$, as by Property~\ref{pr:no-two-edges-same-structure} this would imply that $P^i_{\beta_\ell}$ contains $e^{i+1}_{\beta_\ell}$, contradicting the fact that such an edge does not belong to $S^1[\beta_\ell]\cup S^2[\beta_\ell]$.


\item We now deal with any cluster $\mu$ such that there exists a con-edge $h^i_j$ for $\mu$ in $H_j(e^i_{\alpha})$, for some $j\geq 1$. Observe that $h^i_j\in S$. Denote by $l^i_j$ any con-edge in $L_j(e^i_{\alpha})$ such that $l^i_j \conf h^i_{j}$ and by $h^i_{j-1}$ any con-edge in $H_{j-1}(e^i_{\alpha})$ such that $h^i_{j-1} \conf l^i_j$; these edges exist by definition of conflicting structure and since $j\geq 1$. Denote by $S^1[\mu]$ and $S^2[\mu]$ the two connected components of $S[\mu]$ obtained by removing $h^i_j$ from $S[\mu]$. In the following we prove that the end-vertices of $\delta(h^i_j)$ are one in $S^1[\mu]$ and the other one in $S^2[\mu]$.

    Denote by $\varrho$ and $\tau$ the clusters $l^i_j$ and $h^i_{j-1}$ are con-edges for. Since $e^i_\alpha$ and $e^{i+1}_\alpha$ have isomorphic conflicting structures, we have that $\delta(h^i_j)$, $\delta(l^i_j)$, and $\delta(h^i_{j-1})$ are con-edges for $\mu$, $\varrho$, and $\tau$, respectively. Also, by assumption, $h^i_j$, $h^i_{j-1}$, and $\delta(l^i_{j})$ belong to $S$, while $l^i_j$, $\delta(h^i_{j-1})$, and $\delta(h^i_{j})$ do not.

    Suppose, for a contradiction, that both the end-vertices of $\delta(h^i_j)$ are in $S^1[\mu]$. The end-vertices of $h^i_j$ are one in $S^1[\mu]$ and the other one in $S^2[\mu]$, given that $S[\mu]$ is a tree and that $S^1[\mu]$ and $S^2[\mu]$ are obtained from $S[\mu]$ by removing edge $h^i_j$. Refer to Fig.~\ref{fig:distinct-components-3}.

    First, consider the path $P^{i+1}_{\mu}$ connecting the end-vertices of $\delta(h^i_j)$ and all of whose edges belong to $S^1[\mu]$. By assumption, $S^1[\mu]$ (and hence $P^{i+1}_{\mu}$) does not contain $h^i_j$ nor $\delta(h^i_{j})$. Then, consider the cycle ${\cal C}^{i+1}_{\mu}$ composed of $P^{i+1}_{\mu}$ and of $\delta(h^i_j)$. We have that con-edges $h^i_{j-1}$ and $\delta(h^i_{j-1})$ for $\tau$ do not cross ${\cal C}^{i+1}_{\mu}$ by Property~\ref{pr:no-two-edges-same-structure}; in fact, $h^i_{j-1}$ ($\delta(h^i_{j-1})$) belongs to the same connected component of $K_A$ as $h^i_{j}$ (resp.\  $\delta(h^i_{j})$) and it does not cross $h^i_{j}$ (resp.\ $\delta(h^i_{j})$), as otherwise {\sc Test 2} would succeed on $A$, hence it does not cross any con-edge for $\mu$. Also, we have that no con-edge $g_{\tau}$ for $\tau$ in $S$ crosses ${\cal C}^{i+1}_{\mu}$. In fact, $g_{\tau}$ cannot cross $P^{i+1}_{\mu}$, as such a path is composed of con-edges in $S$, and it cannot cross $\delta(h^i_{j})$ by Property~\ref{pr:no-two-edges-same-structure}, since $\delta(h^i_{j})$ and con-edge $\delta(h^i_{j-1})$ for $\tau$ belong to the same connected component of $K_A$ and do not cross. These observations immediately lead to a contradiction in the case in which $h^i_{j-1}$ and $\delta(h^i_{j-1})$ are on different sides of ${\cal C}^{i+1}_{\mu}$, as in such a case no path whose edges belong to $S$ can connect an end-vertex of $h^i_{j-1}$ with an end-vertex of $\delta(h^i_{j-1})$ without crossing ${\cal C}^{i+1}_{\mu}$. Hence, assume that $h^i_{j-1}$ and $\delta(h^i_{j-1})$ are on the same side of ${\cal C}^{i+1}_{\mu}$ (call it {\em the large side} of ${\cal C}^{i+1}_{\mu}$). Call {\em the small side} of ${\cal C}^{i+1}_{\mu}$ the side of ${\cal C}^{i+1}_{\mu}$ that does not contain $h^i_{j-1}$ and $\delta(h^i_{j-1})$.

\begin{figure}[tb] 
\begin{center}
\begin{tabular}{c c}
\mbox{\includegraphics[scale=0.38]{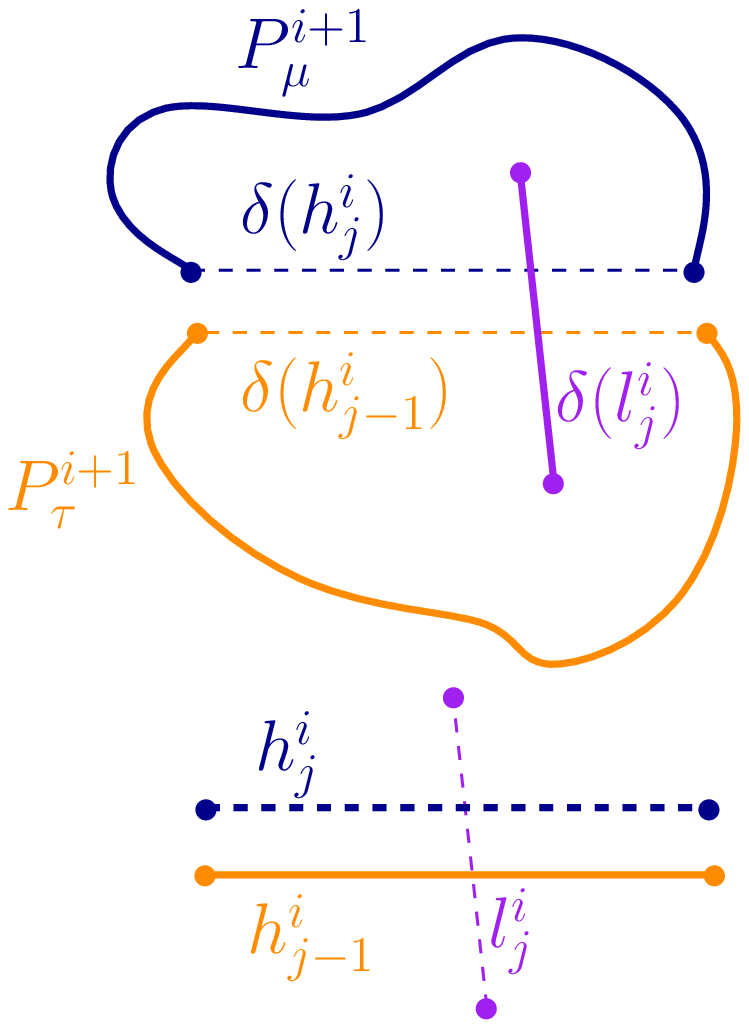}} \hspace{3mm} &
\mbox{\includegraphics[scale=0.38]{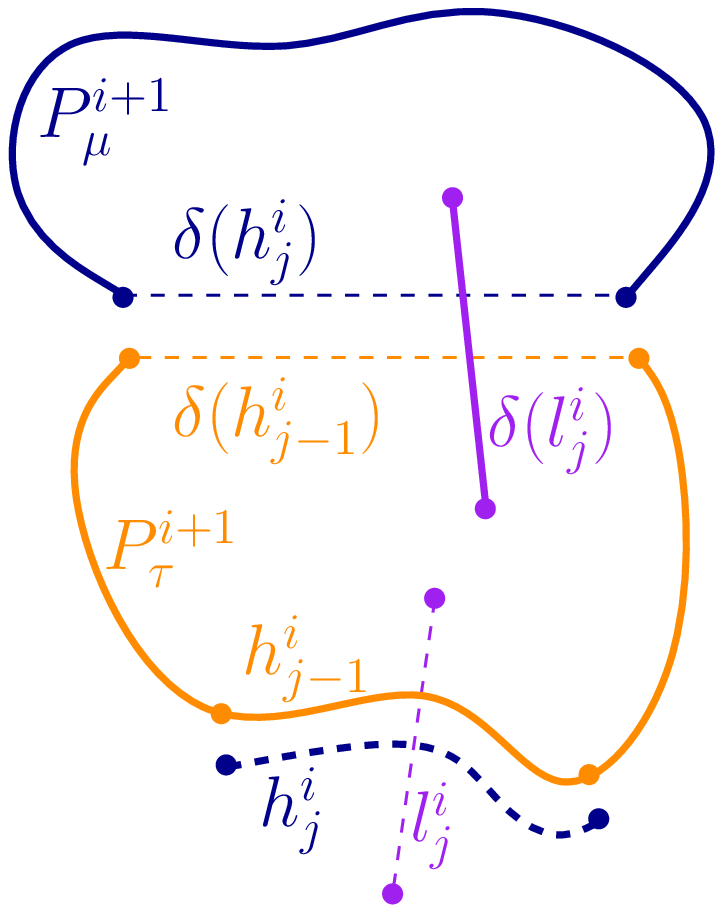}}\\
(a) \hspace{3mm} & (b)
\end{tabular}
\caption{Illustration for the proof that the end-vertices of $\delta(h^i_j)$ are one in $S^1[\mu]$ and the other one in $S^2[\mu]$.}
\label{fig:distinct-components-3}
\end{center}
\end{figure}

Next, consider the path $P^{i+1}_{\tau}$ connecting the end-vertices of $\delta(h^i_{j-1})$ and all of whose edges belong to $S$. Observe that $P^{i+1}_{\tau}$ does not coincide with $\delta(h^i_{j-1})$, given that $\delta(h^i_{j-1})\notin S$. (Observe that $P^{i+1}_{\tau}$ can possibly contain edge $h^i_{j-1}$, asymmetrically to  $P^{i+1}_{\mu}$ that does not contain $h^i_{j}$.)  Then, consider the cycle ${\cal C}^{i+1}_{\tau}$ composed of $P^{i+1}_{\tau}$ and of $\delta(h^i_{j-1})$. Analogously as for ${\cal C}^{i+1}_{\mu}$, it can be concluded that $h^i_{j}$ and $\delta(h^i_{j})$ are on the same side of ${\cal C}^{i+1}_{\tau}$ (call it {\em the large side} of ${\cal C}^{i+1}_{\tau}$). Call {\em the small side} of ${\cal C}^{i+1}_{\tau}$ the side of ${\cal C}^{i+1}_{\tau}$ that does not contain $h^i_{j}$ and $\delta(h^i_{j})$. Hence, ${\cal C}^{i+1}_{\mu}$ is in the large side of ${\cal C}^{i+1}_{\tau}$ and ${\cal C}^{i+1}_{\tau}$ is in the large side of ${\cal C}^{i+1}_{\mu}$, thus the small sides of  ${\cal C}^{i+1}_{\mu}$ and ${\cal C}^{i+1}_{\tau}$ have disjoint interiors.

Now, consider edge $\delta(l^i_j)$. Since it crosses $\delta(h^i_{j})$, then one of its end-vertices is in the small side of ${\cal C}^{i+1}_{\mu}$; also, since it crosses $\delta(h^i_{j-1})$, the other end-vertex is in the small side of ${\cal C}^{i+1}_{\tau}$. Hence, in order to obtain a contradiction, it suffices to prove that there exists a vertex $v_{\varrho}$ of $\varrho$ that is neither in the small side of  ${\cal C}^{i+1}_{\mu}$ nor in the small side of ${\cal C}^{i+1}_{\tau}$ (that is, $v_{\varrho}$ is simultaneously in the large side of  ${\cal C}^{i+1}_{\mu}$ and in the large side of  ${\cal C}^{i+1}_{\tau}$). In fact, if that is the case, then no path whose edges belong to $S$ can connect $v_{\varrho}$ with the end-vertices of $\delta(l^i_j)$, given that no con-edge for $\varrho$ in $S$ can cross an edge of $P^{i+1}_{\mu}\cup P^{i+1}_{\tau}$, hence $S$ does not connect $\varrho$, a contradiction.

We claim that at least one of the end-vertices of $l^i_j$ is simultaneously in the large side of  ${\cal C}^{i+1}_{\mu}$ and in the large side of  ${\cal C}^{i+1}_{\tau}$. First, observe that both the end-vertices of $l^i_j$ are in the large side of ${\cal C}^{i+1}_{\mu}$. In fact, $h^i_{j-1}$ is in the large side of ${\cal C}^{i+1}_{\mu}$, by assumption, and hence all of $h^i_{j-1}$, $h^i_{j}$, and $l^i_{j}$ are in the large side of ${\cal C}^{i+1}_{\mu}$, given that ${\cal C}^{i+1}_{\mu}$ does not contain $h^i_j$. Analogously, if ${\cal C}^{i+1}_{\tau}$ does not contain $h^i_{j-1}$ (as in Fig.~\ref{fig:distinct-components-3}(a)), then all of $h^i_{j-1}$, $h^i_{j}$, and $l^i_{j}$ are in the large side of ${\cal C}^{i+1}_{\tau}$; on the other hand, if ${\cal C}^{i+1}_{\tau}$ contains $h^i_{j-1}$ (as in Fig.~\ref{fig:distinct-components-3}(b)), then $l^i_{j}$ crosses ${\cal C}^{i+1}_{\tau}$, hence one of its end-vertices is in the small side of ${\cal C}^{i+1}_{\tau}$ and the other end-vertex is in the large side of ${\cal C}^{i+1}_{\tau}$. This proves the claim and hence the statement.


\item It remains to deal with any cluster $\mu$ such that there exists a con-edge $l^i_j$ for $\mu$ in $L_j(e^i_{\alpha})$, for some $j\geq 2$. Observe that $l^i_j\notin S$, while $\delta(l^i_j)\in S$. Denote by $h^i_{j-1}$ any con-edge in $H_{j-1}(e^i_{\alpha})$ such that $h^i_{j-1} \conf l^i_{j}$ and by $l^i_{j-1}$ any con-edge in $L_{j-1}(e^i_{\alpha})$ such that $l^i_{j-1} \conf h^i_{j-1}$. All these edges exist by definition of conflicting structure and since $j\geq 2$. Denote by $S^1[\mu]$ and $S^2[\mu]$ the two connected components of $S[\mu]$ obtained by removing $\delta(l^i_j)$ from $S[\mu]$. The following statement can be proved: The end-vertices of $l^i_j$ are one in $S^1[\mu]$ and the other one in $S^2[\mu]$. The proof is the same as the one for the case in which there exists a con-edge $h^i_j\in H_j(e^i_{\alpha})$ for $\mu$ in $S$, for some $j\geq 1$, with $l^i_j$, $h^i_{j-1}$, $l^i_{j-1}$, $\delta(l^i_j)$, $\delta(h^i_{j-1})$, and $\delta(l^i_{j-1})$ playing the roles of $\delta(h^i_{j})$, $\delta(l^i_j)$, $\delta(h^i_{j-1})$, $h^i_{j}$, $l^i_j$, and $h^i_{j-1}$, respectively.
\end{itemize}

This concludes the proof of the lemma. Together with Lemma~\ref{le:one-or-the-other}, it establishes the correctness of {\sc Simplification 5}.
\end{proof}


Next, we study non-isomorphic spokes. Let $e^i_\alpha$ be a spoke of the $\alpha$-donut for $e_{\alpha}$. Assume that $L_1(e^i_\alpha)$ contains a con-edge $e^i_{\beta}$ for a cluster $\beta$, and that $H_1(e^i_\alpha)$ contains a con-edge $e^i_{\gamma}$ for a cluster $\gamma$, where $e^i_\alpha\conf e^i_{\beta}$ and $e^i_{\beta} \conf e^i_{\gamma}$. By Property~\ref{pr:no-two-edges-same-structure}, since $e^i_\gamma$ and $e^i_\alpha$ belong to the same connected component of $K_A$ and do not cross (as otherwise {\sc Test 2} would succeed), it follows that $e^i_\gamma$ does not cross any con-edge for $\alpha$, hence it lies in one of the two faces $f^{i}_\alpha$ and $f^{i+1}_\alpha$ of $A[\alpha]$ that $e^i_\alpha$ shares with spokes $e^{i-1}_\alpha$ and $e^{i+1}_\alpha$, respectively. Assume w.l.o.g. that $e^i_\gamma$ lies in $f^{i+1}_\alpha$. By Lemma~\ref{le:donut}, $L_1(e^{i+1}_\alpha)$ contains a con-edge $e^{i+1}_{\beta}$ for $\beta$, where $e^{i+1}_\alpha\conf e^{i+1}_{\beta}$.

The next two lemmata discuss the case in which $M(e^{i+1}_\alpha)$ contains a con-edge for $\gamma$ that has a conflict with $e^{i+1}_{\beta}$ and the case in which it does not. We start with the latter.

\begin{lemma}[{\sc Simplification 6}] \label{le:non-isomorphic-differentT1}
Suppose that no con-edge $e^{i+1}_{\gamma}$ for $\gamma$ exists such that $e^{i+1}_{\gamma} \conf e^{i+1}_{\beta}$, and that a planar set $S$ of spanning trees for $A$ exists. Then, $e^i_\alpha\in S$.
\end{lemma}

\begin{proof}
Suppose, for a contradiction, that a planar set $S$ of spanning trees for $A$ exists with  $e^i_\alpha \notin S$. Since no two conflicting edges both belong to $S$ and by Lemma~\ref{le:one-or-the-other}, we have $e^i_\beta \in S$ and $e^i_\gamma \notin S$. Refer to Fig.~\ref{fig:non-isomorphic-one}.

\begin{figure}[tb]
\begin{center}
\mbox{\includegraphics[scale=0.38]{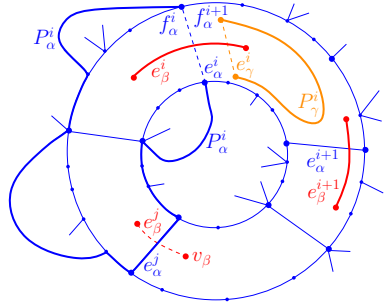}}
\caption{Illustration for the proof of Lemma~\ref{le:non-isomorphic-differentT1}.}
\label{fig:non-isomorphic-one}
\end{center}
\end{figure}

Consider the path $P^{i}_{\gamma}$ connecting the end-vertices of $e^i_\gamma$ and all of whose edges belong to $S$. By assumption, $P^{i}_{\gamma}$ does not coincide with $e^i_\gamma$. Then, consider the cycle ${\cal C}^{i}_{\gamma}$ composed of $P^{i}_{\gamma}$ and of $e^i_\gamma$. We have that ${\cal C}^{i}_{\gamma}$ does not cross any con-edge for $\alpha$ (including those not in $S$). Indeed, suppose that ${\cal C}^{i}_{\gamma}$ crosses a con-edge $e^*_{\alpha}$ for $\alpha$. Then, it contains vertices of $\alpha$ on both sides (e.g., the end-vertices of $e^*_{\alpha}$). However, no con-edge $g_{\alpha}$ for $\alpha$ in $S$ crosses ${\cal C}^{i}_{\gamma}$. In fact, $g_{\alpha}$ cannot cross $P^{i}_{\gamma}$, as such a path is composed of con-edges in $S$, and it cannot cross $e^i_\gamma$, as $e^i_\gamma$ and $e^i_\alpha$ belong to the same connected component of $K_A$ and do not cross; it follows that $S$ does not connect $\alpha$, a contradiction. We can hence assume that all the edges of ${\cal C}^{i}_{\gamma}$ lie inside a single face of $A[\alpha]$. Since $e^i_\gamma$ lies inside $f^{i+1}_{\alpha}$, we have that all the edges of ${\cal C}^{i}_{\gamma}$ lie inside $f^{i+1}_{\alpha}$. We emphasize that ${\cal C}^{i}_{\gamma}$ does not cross $e^{i+1}_{\beta}$, given that the latter crosses no con-edge for $\gamma$, by assumption.

Next, consider the path $P^{i}_{\alpha}$ connecting the end-vertices of $e^i_\alpha$ and all of whose edges belong to $S$. By assumption, $P^{i}_{\alpha}$ does not coincide with $e^i_\alpha$. By Lemma~\ref{le:exactly-one} and since $e^1_\alpha,\dots,e^{k}_\alpha$ form a separating set for $A[\alpha]$, we have that $P^{i}_{\alpha}$ contains exactly one of $e^1_\alpha,\dots,e^{k}_\alpha$, say $e^j_\alpha$ with $j\neq i$. Denote by $e^j_\beta$ the con-edge for $\beta$ that has a conflict with $e^j_\alpha$. This edge exists by Lemma~\ref{le:donut}. Also, denote by ${\cal C}^{i}_{\alpha}$ the cycle $P^{i}_{\alpha}\cup e^i_\alpha$.

Cycles ${\cal C}^{i}_{\alpha}$ and ${\cal C}^{i}_{\gamma}$ do not cross, since ${\cal C}^{i}_{\gamma}$ lies inside $f^{i+1}_{\alpha}$. Then, call the {\em large side} of ${\cal C}^{i}_{\alpha}$ the side that contains all the edges of ${\cal C}^{i}_{\gamma}$ (call the other side of ${\cal C}^{i}_{\alpha}$ its {\em small side}); also call the {\em large side} of ${\cal C}^{i}_{\gamma}$ the side that contains all the edges of ${\cal C}^{i}_{\alpha}$ (call the other side of ${\cal C}^{i}_{\gamma}$ its {\em small side}). Thus, thus the small sides of ${\cal C}^{i}_{\gamma}$ and ${\cal C}^{i}_{\alpha}$ have disjoint interiors.

Now, consider edge $e^i_{\beta}$. Since it crosses $e^i_{\alpha}$, one of its end-vertices is in the small side of ${\cal C}^{i}_{\alpha}$; also, since it crosses $e^{i}_{\gamma}$, the other end-vertex is in the small side of ${\cal C}^{i}_{\gamma}$. Hence, in order to obtain a contradiction, it suffices to prove that there exists a vertex $v_{\beta}$ of $\beta$ that is simultaneously in the large side of ${\cal C}^{i}_{\alpha}$ and in the large side of ${\cal C}^{i}_{\gamma}$. In fact, if that is the case, then no path whose edges belong to $S$ can connect $v_{\beta}$ with the end-vertices of $e^i_{\beta}$, given that no con-edge for $\beta$ in $S$ can cross an edge of $P^{i}_{\alpha} \cup P^{i}_{\gamma}$; thus, $S$ does not connect $\beta$, a contradiction.

We claim that one of the end-vertices of $e^j_\beta$ is simultaneously in the large side of ${\cal C}^{i}_{\alpha}$ and in the large side of ${\cal C}^{i}_{\gamma}$. First, we prove that both the end-vertices of $e^j_\beta$ are in the large side of ${\cal C}^{i}_{\gamma}$. Namely, since ${\cal C}^{i}_{\gamma}$ does not cross any con-edge for $\alpha$, the end-vertices of $e^j_\alpha$ are both in the large side of ${\cal C}^{i}_{\gamma}$. Hence, if one of the end-vertices of $e^j_\beta$ is not in the large side of ${\cal C}^{i}_{\gamma}$, it follows that $e^j_\beta$ crosses $P^{i}_{\gamma}$. Since all the edges of $P^{i}_{\gamma}$ lie in $f^{i+1}_{\alpha}$, we have that $e^j_\beta$ crosses $P^{i}_{\gamma}$ only if $j=i+1$. However, this contradicts the assumption that $e^{i+1}_\beta$ does not cross any con-edge for $\gamma$. Second, since $e^j_\beta$ crosses $P^{i}_{\alpha}$, one of its end-vertices is in the small side of ${\cal C}^{i}_{\alpha}$, while the other one, say $v_{\beta}$, is in the large side of ${\cal C}^{i}_{\alpha}$. Hence, $v_{\beta}$ is in the large side of both ${\cal C}^{i}_{\alpha}$ and ${\cal C}^{i}_{\gamma}$. This proves the claim and hence the lemma. Together with Lemma~\ref{le:one-or-the-other}, this lemma establishes the correctness of {\sc Simplification~6}.
\end{proof}



\begin{lemma}[{\sc Simplification 7}]\label{le:non-isomorphic-sameT1}
Suppose that a con-edge $e^{i+1}_\gamma$ for $\gamma$ exists with $e^{i+1}_{\gamma}\conf e^{i+1}_{\beta}$. If a planar set $S$ of spanning trees for $A$ exists, then either $e^i_\alpha\in S$ or $e^{i+1}_\alpha\in S$.
\end{lemma}

\begin{proof}
Suppose that a planar set $S$ of spanning trees for $A$ exists. By Lemma~\ref{le:exactly-one}, {\em exactly} one out of $e^1_\alpha,\dots,e^{k}_\alpha$ belongs to $S$. Hence, {\em at most} one out of $e^i_\alpha$ and $e^{i+1}_\alpha$ belongs to $S$. It remains to prove that {\em at least} one out of $e^i_\alpha$ and $e^{i+1}_\alpha$ belongs to $S$. By Lemma~\ref{le:one-or-the-other}, this is equivalent to prove that {\em at most} one out of $e^i_\beta$ and $e^{i+1}_\beta$ belongs to $S$.

We prove that $e^{i+1}_\beta$ is a spoke of the $\beta$-donut for $e^i_\beta$. By Lemma~\ref{le:exactly-one}, the statement implies that at most one out of $e^i_\beta$ and $e^{i+1}_\beta$ belongs to $S$, and hence implies the lemma.

\begin{figure}[tb]
\begin{center}
\mbox{\includegraphics[scale=0.38]{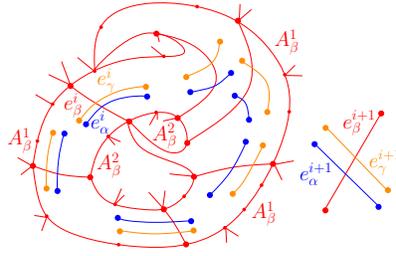}}
\caption{Illustration for the proof of Lemma~\ref{le:non-isomorphic-sameT1}.}
\label{fig:non-isomorphic-two}
\end{center}
\end{figure}

Suppose, for a contradiction, that $e^{i+1}_\beta$ is not a spoke of the $\beta$-donut for $e^i_\beta$. Refer to Fig.~\ref{fig:non-isomorphic-two}. First, the $\beta$-donut for $e^i_\beta$ exists by Lemma~\ref{le:donut}, given that $e^i_\beta$ crosses at least two con-edges $e^i_\alpha$ and $e^i_{\gamma}$ for clusters $\alpha$ and $\gamma$, respectively, and given that {\sc Simplifications 1--4} do not apply to $A$ and {\sc Tests 1--4} fail on $A$. Second, denote by $\tau_1,\dots,\tau_m$ the clusters whose con-edges cross $e^i_\beta$, ordered as they cross $e^i_\beta$ when clockwise traversing one of the two faces incident to $e^i_\beta$. Observe that $\alpha$ and $\gamma$ are among $\tau_1,\dots,\tau_m$. Third, we define two subgraphs $A^{1}_\beta$ and $A^{2}_\beta$ of $A[\beta]$, as the subgraphs of $A[\beta]$ whose edges delimit the $\beta$-donut for $e^i_\beta$. That is, consider the faces of $A[\beta]$ incident to spokes of the $\beta$-donut for $e^i_\beta$; the union of the boundaries of such faces defines a connected subgraph of $A[\beta]$, from which we remove the spokes of the $\beta$-donut for $e^i_\beta$, thus obtaining a subgraph $A^*_{\beta}$ of $A[\beta]$ composed of two connected components, that we denote by $A^{1}_\beta$ and $A^{2}_\beta$. By Lemma~\ref{le:donut}, the edges of $A^{1}_\beta$ are not crossed by any con-edge for $\tau_2,\dots,\tau_m$, and the edges of $A^{2}_\beta$ are not crossed by any con-edge for $\tau_1,\dots,\tau_{m-1}$ (up to renaming $A^{1}_\beta$ with $A^{2}_\beta$). Denote by $f^A_{\beta}$ the connected region defined by $A^*_{\beta}$ that used to contain the spokes of the $\beta$-donut for $e^i_\beta$.

If $e^{i+1}_\beta$ is not a spoke of the $\beta$-donut for $e^i_\beta$, then either $A^{1}_\beta$ or $A^{2}_\beta$ separates $e^{i+1}_\beta$ from $f^A_{\beta}$, given that the only edges of $A[\beta]$ in $f^A_{\beta}$ are the spokes of the $\beta$-donut for $e^i_\beta$. Suppose w.l.o.g. that $A^{1}_\beta$ separates $e^{i+1}_\beta$ from $f^A_{\beta}$. Observe that at least one of $\alpha$ and $\gamma$ is in $\tau_2,\dots,\tau_m$, say that $\alpha$ is in $\tau_2,\dots,\tau_m$. Then, either $A[\alpha]$ is disconnected, or there exists a path that is composed of con-edges for $\alpha$, that connects an end-vertex of $e^{i+1}_\alpha$ with an end-vertex of $e^{i}_\alpha$, and that contains an edge crossing an edge of $A^{1}_\beta$. In both cases we have a contradiction, which proves the statement and hence the lemma. Together with Lemma~\ref{le:one-or-the-other}, this lemma establishes the correctness of {\sc Simplification~7}.
\end{proof}

Observe that Simplification 7 can be applied in the case in which the $\alpha$-donut for $e_{\alpha}$ has at least three spokes. Namely, in that case, by Lemmata~\ref{le:exactly-one} and~\ref{le:non-isomorphic-sameT1} all the spokes different from $e^i_\alpha$ and $e^{i+1}_\alpha$ can be removed from $A$.



Next, assume that there exists an $\alpha$-donut with exactly two spokes $e^1_\alpha$ and $e^2_\alpha$. Consider the smallest $j\geq 1$ such that one of the following holds:

\begin{enumerate}
\item there exist con-edges $e_{\mu}\in L_j(e^a_{\alpha})$ and $e_{\nu}\in H_{j-1}(e^a_{\alpha})$ for clusters $\mu$ and $\nu$, resp., such that $e_{\mu}\conf e_{\nu}$, and there exists no con-edge $g_{\mu}\in L_j(e^b_{\alpha})$ for $\mu$ such that $g_{\mu}\conf g_{\nu}$ with $g_{\nu}$ con-edge for $\nu$ in $H_{j-1}(e^b_{\alpha})$, for some $a,b\in\{1,2\}$ with $a\neq b$; or
\item there exist con-edges $e_{\mu}\in H_j(e^a_{\alpha})$ and $e_{\nu}\in L_{j}(e^a_{\alpha})$ for clusters $\mu$ and $\nu$, resp., such that $e_{\mu}\conf e_{\nu}$, and there exists no con-edge $g_{\mu}\in H_j(e^b_{\alpha})$ for $\mu$ such that $g_{\mu}\conf g_{\nu}$ with $g_{\nu}$ con-edge for $\nu$ in $L_{j}(e^b_{\alpha})$, for some $a,b\in\{1,2\}$ with $a\neq b$.
\end{enumerate}

We have the following.

\begin{lemma}[{\sc Simplification 8}] \label{le:non-isomorphic-differentTj}
Assume that a planar set $S$ of spanning trees for $A$ exists. Then, $e_{\mu}\in S$.
\end{lemma}

\begin{proof}
We prove the lemma in the case in which $j$ is determined by (1), i.e. there exist con-edges $e_{\mu}\in L_j(e^a_{\alpha})$ and $e_{\nu}\in H_{j-1}(e^a_{\alpha})$ for clusters $\mu$ and $\nu$, respectively, such that $e_{\mu}\conf e_{\nu}$, and there exists no con-edge $g_{\mu}\in L_j(e^b_{\alpha})$ for $\mu$ such that $g_{\mu}\conf g_{\nu}$ with $g_{\nu}\in H_{j-1}(e^b_{\alpha})$ con-edge for $\nu$, for some $a,b\in\{1,2\}$ with $a\neq b$. The proof for the case in which the value of $j$ is determined by (2) is analogous. Refer to Fig.~\ref{fig:two-spokes}.

Let $e^h_0=e^a_{\alpha}$, $g^h_0=e^b_{\alpha}$, and $e^l_j=e_{\mu}$. Also, for $1\leq k\leq j-1$, denote by $e^l_k\in L_k(e^a_{\alpha})$ and $e^h_k \in H_k(e^a_{\alpha})$ con-edges for clusters $\nu^l_k$ and $\nu^h_k$, respectively, where $e^h_k\conf e^l_{k+1}$ for $0\leq k\leq j-1$ and $e^l_k\conf e^h_k$ for $1\leq k\leq j-1$. Observe that $e^h_{j-1}=e_{\nu}$. Further, for $1\leq k\leq j-1$, denote by $g^l_k\in L_k(e^b_{\alpha})$ and $g^h_k \in H_k(e^b_{\alpha})$ con-edges for clusters $\nu^l_k$ and $\nu^h_k$, respectively, such that $g^h_k\conf g^l_{k+1}$ for $0\leq k\leq j-2$ and $g^l_k\conf g^h_k$ for $1\leq k\leq j-1$. All these edges exist by definition of conflicting structure and by the minimality of $j$. Observe that $g^h_{j-1}=g_{\nu}$.  By assumption, no con-edge $g_{\mu}\in L_j(e^b_{\alpha})$ for $\mu$ exists such that $g_{\mu}\conf g^h_{j-1}$.

\begin{figure}[tb]
\begin{center}
\mbox{\includegraphics[scale=0.38]{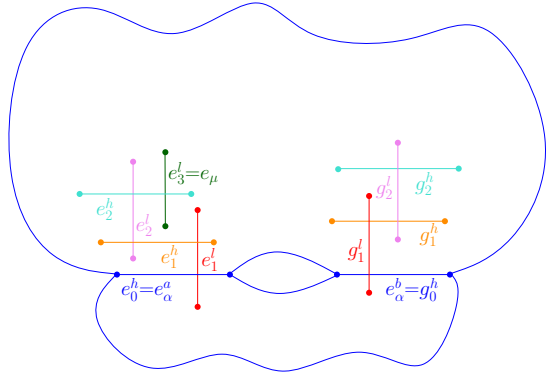}}
\caption{Illustration for the proof of Lemma~\ref{le:non-isomorphic-differentTj}, with $j=3$.}
\label{fig:two-spokes}
\end{center}
\end{figure}

First, we argue that no con-edge $g_{\mu}$ for $\mu$ exists such that $g_{\mu}\conf g^h_{j-1}$. That is, not only $g_{\mu}$ is not in $L_j(e^b_{\alpha})$, but no con-edge $g_{\mu}$ for $\mu$ such that $g_{\mu}\conf g^h_{j-1}$ exists in $M(e^b_{\alpha})$ at all. By definition of conflicting structure and since {\sc Test 2} does not apply to $A$, if $g_{\mu}\in M(e^b_{\alpha})$ and $g_{\mu}\notin L_j(e^b_{\alpha})$, then $g_{\mu}\in L_{j-1}(e^b_{\alpha})$. However, this contradicts the minimality of $j$. Namely,  there exist con-edges $g_{\mu}\in L_{j-1}(e^b_{\alpha})$ and $g_{\tau}\in H_{j-2}(e^b_{\alpha})$ for clusters $\mu$ and $\tau$, respectively, such that $g_{\mu}\conf g_{\tau}$ (in fact, $g_{\tau}$ is any edge in $H_{j-2}(e^b_{\alpha})$ that crosses $g_{\mu}$; this edge exists by definition of conflicting structure), and there exists no con-edge $l_{\mu}\in L_{j-1}(e^a_{\alpha})$ for $\mu$ such that $l_{\mu}\conf l_{\tau}$ with $l_{\tau}\in H_{j-2}(e^a_{\alpha})$ con-edge for $\tau$, since by Property~\ref{pr:no-two-edges-same-structure} no con-edge for $\mu$ different from $e_{\mu}$ belongs to $M(e^a_{\alpha})$.

Now suppose, for a contradiction, that a planar set $S$ of spanning trees for $A$ exists with $e_{\mu} \notin S$. Since no two conflicting edges both belong to $S$ and by Lemma~\ref{le:one-or-the-other}, we have that $e^h_0,e^h_1,\dots,e^h_{j-1}\in S$ and $e^l_1,e^l_2,\dots,e^l_{j}\notin S$. Further, by Lemma~\ref{le:exactly-one} we have $e^b_{\alpha} \notin S$. Since no two conflicting edges both belong to $S$ and by Lemma~\ref{le:one-or-the-other}, we have that $g^l_1,\dots,g^l_{j-1}\in S$ and $g^h_0,g^h_1,\dots,g^h_{j-1}\notin S$.

For each $1\leq k\leq j-1$, denote by $P^l_k$ the path connecting the end-vertices of $e^l_k$ and all of whose edges belong to $S$. By assumption, $P^l_k\neq e^l_k$. Hence, denote by ${\cal C}^l_k$ the cycle composed of $P^l_k$ and $e^l_k$. Analogously, for each $0\leq k\leq j-1$, denote by $P^h_k$ the path connecting the end-vertices of $g^h_k$ and all of whose edges belong to $S$. By assumption, $P^h_k\neq g^h_k$. Hence, denote by ${\cal C}^h_k$ the cycle composed of $P^h_k$ and $g^h_k$.

We will iteratively prove the following statements (see  Fig.~\ref{fig:two-spokes-2}): (I) For every $0\leq k\leq j-1$, edge $e^h_k$ belongs to ${\cal C}^h_k$; (II) for every $1\leq k\leq j-1$, edge $g^l_k$ belongs to ${\cal C}^l_k$.

\begin{figure}[tb]
\begin{center}
\mbox{\includegraphics[scale=0.38]{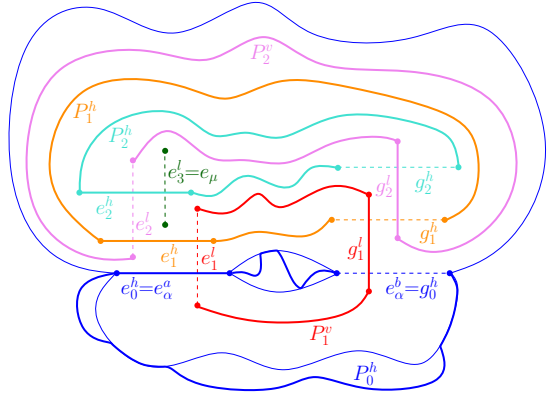}}
\caption{Illustration for statements (I) and (II).}
\label{fig:two-spokes-2}
\end{center}
\end{figure}

First, $e^h_0=e^a_{\alpha}$ belongs to ${\cal C}^h_0$, as a consequence of the fact that $e^a_{\alpha}$ and $e^b_{\alpha}$ form a separating pair of edges for $A[\alpha]$ and that the end-vertices of $e^b_{\alpha}$ are in different connected components resulting from the removal of $e^a_{\alpha}$ and $e^b_{\alpha}$ from $A[\alpha]$.

Now suppose that, for some $0\leq k\leq j-2$, it holds that $e^h_k\in {\cal C}^h_k$; we prove that $g^l_{k+1}\in {\cal C}^l_{k+1}$. Observe that the end-vertices of $e^l_{k+1}$ are on different sides of ${\cal C}^h_k$, given that $e^l_{k+1}\conf e^h_k$ and that $e^h_k\in {\cal C}^h_k$. By the Jordan curve theorem, path $P^l_{k+1}$ crosses an edge of ${\cal C}^h_k$. However, $P^l_{k+1}$ does not cross any edge of $P^h_k$ as all the edges of such paths belong to $S$. Hence, $P^l_{k+1}$ crosses $g^h_k$, hence it contains edge $g^l_{k+1}$, which proves the statement.

Analogously, suppose that, for some $1\leq k\leq j-1$, it holds that $g^l_k\in {\cal C}^l_k$; we prove that $e^h_{k}\in {\cal C}^h_{k}$. Observe that the end-vertices of $g^h_k$ are on different sides of ${\cal C}^l_k$, given that $g^h_k\conf g^l_k$ and that $g^l_k\in {\cal C}^l_k$. By the Jordan curve theorem, path $P^h_k$ crosses an edge of ${\cal C}^l_k$. However, $P^h_k$ does not cross any edge of $P^l_k$ as all the edges of such paths belong to $S$. Hence, $P^h_k$ crosses $e^l_k$, and therefore it contains edge $e^h_k$, which proves the statement.

This proves statements (I) and (II). Now observe that the end-vertices of $e_{\mu}$ are on different sides of ${\cal C}^h_{j-1}$, given that $e_{\mu}$ has a conflict with $e^h_{j-1}$ and that $e^h_{j-1}$ belongs to ${\cal C}^h_{j-1}$. Consider a path $P_{\mu}$ that connects the end-vertices of $e_{\mu}$ and all of whose edges belong to $S$. Since the end-vertices of $e_{\mu}$ are on different sides of ${\cal C}^h_{j-1}$, by the Jordan curve theorem $P_{\mu}$ crosses an edge of ${\cal C}^h_{j-1}$. However, $P_{\mu}$ does not cross any edge of $P^h_{j-1}$ as all the edges of such paths belong to $S$. Hence, $P_{\mu}$ crosses $g^h_{j-1}$, a contradiction to the assumption that no con-edge for $\mu$ crosses $g^h_{j-1}$. This proves the lemma. Together with Lemma~\ref{le:one-or-the-other}, this lemma establishes the correctness of {\sc Simplification~8}.
\end{proof}

We now prove that our simplifications form a ``complete set''.
\begin{lemma} \label{le:no-simplification-single-conflict}
Suppose that {\sc Simplifications 1--8} do not apply to $A$ and that {\sc Tests 1--4} fail on $A$. Then, every con-edge in $A$ crosses exactly one con-edge in $A$.
\end{lemma}

\begin{proof}
Since {\sc Simplification 2} and {\sc Simplification 3} do not apply to $A$, every con-edge in $A$ has a conflict with at least one con-edge in $A$. Hence, we need to prove that there exists no con-edge in $A$ that has a conflict with two or more con-edges in $A$. Suppose, for a contradiction, that there exists a con-edge $e_{\alpha}\in A$ for a cluster $\alpha$ that has a conflict with con-edges for clusters $\beta_1,\dots,\beta_m$, for some $m\geq 2$.

Since {\sc Simplifications 1--4} do not apply to $A$ and {\sc Tests 1--4} fail on $A$, by Lemma~\ref{le:donut} there exists an $\alpha$-donut $D_{\alpha}$ having $e_{\alpha}$ as one of its spokes.

Suppose first that $D_{\alpha}$ has more than two spokes. Consider any two spokes $e^i_{\alpha}$ and $e^{i+1}_{\alpha}$ in $D_{\alpha}$. If $H_1(e^i_{\alpha})=H_1(e^{i+1}_{\alpha})=\emptyset$, then $e^i_{\alpha}$ and $e^{i+1}_{\alpha}$ have isomorphic conflicting structures, hence {\sc Simplification 5} applies to $A$, a contradiction. Hence, $H_1(e^i_{\alpha})\neq \emptyset$ or $H_1(e^{i+1}_{\alpha})\neq \emptyset$, say w.l.o.g. $H_1(e^i_{\alpha})\neq \emptyset$. Let $e^i_{\gamma}\in H_1(e^i_{\alpha})$ be any con-edge for a cluster $\gamma$, and let $e^i_{\beta_\ell}\in L_1(e^i_{\alpha})$ be a con-edge for a cluster $\beta_\ell$, for some $1\leq l\leq m$, such that $e^i_{\gamma}\conf e^i_{\beta_\ell}$. Also, $e^i_{\beta_\ell} \conf e^i_{\alpha}$, given that $e^i_{\beta_\ell}\in L_1(e^i_{\alpha})$. Since $e^i_{\gamma}$ and $e^i_{\alpha}$ belong to the same connected component of $K_A$ and do not cross, by Property~\ref{pr:no-two-edges-same-structure}, it follows that $e^i_{\gamma}$ does not cross any con-edge for $\alpha$, hence it lies either in the face $f^{i+1}_{\alpha}$ shared by $e^i_{\alpha}$ and $e^{i+1}_{\alpha}$ or in the face $f^{i}_{\alpha}$ shared by $e^i_{\alpha}$ and $e^{i-1}_{\alpha}$, say w.l.o.g. that $e^i_{\gamma}$ lies in $f^{i+1}_{\alpha}$. If the con-edge $e^{i+1}_{\beta_\ell}\in L_1(e^{i+1}_{\alpha})$ for $\beta_\ell$ has no conflict with any con-edge for $\gamma$, then {\sc Simplification 6} applies to $A$, while if $e^{i+1}_{\beta_\ell}$ has a conflict with a con-edge for $\gamma$, then {\sc Simplification 7} applies to $A$. In both cases, we get a contradiction to the fact that {\sc Simplifications 1--8} do not apply to $A$.

Suppose next that $D_{\alpha}$ has exactly two spokes $e^1_{\alpha}$ and $e^2_{\alpha}$. If $e^1_{\alpha}$ and $e^2_{\alpha}$ have isomorphic conflicting structures, then {\sc Simplification 5} applies to $A$, a contradiction. Otherwise, consider a minimal index $j$ such that either (1) there exists a con-edge $e_{\mu}\in L_j(e^a_{\alpha})$ for a cluster $\mu$ that crosses a con-edge $e_{\nu}\in H_{j-1}(e^a_{\alpha})$ for a cluster  $\nu$, and there exists no con-edge $g_{\mu}\in L_j(e^b_{\alpha})$ for $\mu$ that crosses a con-edge $g_{\nu}\in H_{j-1}(e^b_{\alpha})$ for $\nu$, for some $a,b\in\{1,2\}$ with $a\neq b$, or (2) there exists a con-edge $e_{\mu}\in H_j(e^a_{\alpha})$ for a cluster $\mu$ that crosses a con-edge $e_{\nu}\in L_{j}(e^a_{\alpha})$ for a cluster  $\nu$, and there exists no con-edge $g_{\mu}\in H_j(e^b_{\alpha})$ for $\mu$ that crosses a con-edge $g_{\nu}\in L_{j}(e^b_{\alpha})$ for $\nu$, for some $a,b\in\{1,2\}$ with $a\neq b$. Observe that (1) or (2) has to apply (as otherwise $e^1_{\alpha}$ and $e^2_{\alpha}$ would have isomorphic conflicting structures). But then {\sc Simplification 8} applies to $A$, a contradiction that proves the lemma.
\end{proof}

A linear-time algorithm to determine whether a planar set $S$ of spanning trees exists for a single-conflict graph is known~\cite{df-ectefcgsf-09}. We thus finally get:

\begin{theorem} \label{th:main}
There exists an $O(|C|^3)$-time algorithm to test the $c$-planarity of an embedded flat clustered graph $C$ with at most two vertices per cluster on the boundary of each face.
\end{theorem}

\begin{proof}
The multigraph $A$ of the con-edges can be easily constructed in $O(|C|^2)$ time, so that $A$ has $O(|C|)$ vertices and edges and satisfies Property~\ref{pr:no-two-edges-same-structure}. By Lemma~\ref{le:pssttm}, it suffices to show how to solve the {\sc pssttm} problem for $A$ in $O(|C|^3)$ time.

Algorithm~1 correctly determines whether a planar set $S$ of spanning trees for $A$ exists, by Lemmata~\ref{le:disconnected}--\ref{le:no-simplification-single-conflict}. By suitably equipping each con-edge $e$ in $A$ with pointers to the edges in $A$ that have a conflict with $e$, it can be easily tested in $O(|A|^2)$ time whether the pre-conditions of each of {\sc Simplifications 1--8} and {\sc Tests 1--4} are satisfied; also, the actual simplifications, that is, removing and contracting edges in $A$, can be performed in $O(|A|)$ time. Furthermore, the algorithm in~\cite{df-ectefcgsf-09} runs in $O(|A|)$ time. Since the number of performed tests and simplifications is in $O(|A|)$, the total running time is $O(|A|^3)$, and hence $O(|C|^3)$.
\end{proof}

\section{Conclusions} \label{se:conclusions}

We presented a polynomial-time algorithm for testing $c$-planarity of embedded flat clustered graphs with at most two vertices per cluster on each face. An interesting extension of our results would be to devise an FPT algorithm to test the $c$-planarity of embedded flat clustered graphs, where the parameter is the maximum number $k$ of vertices of the same cluster on the same face. Even an algorithm with running time $n^{O(f(k))}$ seems to be an elusive goal. Several key lemmata (e.g. Lemmata~\ref{le:one-or-the-other} and~\ref{le:conflicts-are-bipartite}) do not apply if $k>2$, hence a deeper study of the combinatorial properties of embedded flat clustered graphs may be necessary.





\bibliographystyle{splncs_srt}
\bibliography{journal}


\end{document}